\documentclass{amsart}

 \usepackage{microtype}

\usepackage[a4paper, total={6in, 8in}]{geometry}

\usepackage{multirow}

\usepackage{soul}
\usepackage{url}
\usepackage{graphicx}
\usepackage{amsmath}
\usepackage{amsthm}
\usepackage{booktabs}
\usepackage{algorithm}
\usepackage{algorithmic}
\usepackage{ifthen}
\usepackage{xcolor}
\usepackage{todonotes}

\usepackage{latexsym} 
\usepackage{amssymb}
\usepackage{mathtools}
\usepackage{amsmath}
\usepackage{stmaryrd}

 \newtheorem{example}{Example}
 \newtheorem{theorem}{Theorem}
 \newtheorem{definition}{Definition}
 \newtheorem{lemma}{Lemma}
 \newtheorem{corollary}{Corollary}

\usepackage{boxedminipage}

\usepackage{tikz}
\usetikzlibrary{arrows,decorations.pathreplacing}

\usepackage{caption}
\usepackage{subcaption}
\usepackage{xspace}

\newif\ifshort
\newif\iflong

\longtrue
\shortfalse

\newcommand{\eqCSP}{CSP$(\source_e)$}
\newcommand{\beqCSP}{CSP$(\target_e)$}

\newcommand{\sols}{\ensuremath{\mathrm{Sol}}}

\newcommand{\ar}{\ensuremath{\mathrm{ar}}}

\newcommand{\pbDef}[3]{%
\noindent
\begin{center}
\begin{boxedminipage}{0.98 \columnwidth}
#1\\[5pt]
\begin{tabular}{l p{0.70 \columnwidth}}
Input: & #2\\
Question: & #3
\end{tabular}
\end{boxedminipage}
\end{center}
}

\newcommand{\bigoh}{\mathcal{O}}

\newcommand{\cc}[1]{{\mbox{\textnormal{\textsf{#1}}}}\xspace}  %% Complexity class
\newcommand{\NP}{\cc{NP}}

\newcommand{\FPT}{\cc{FPT}}

\newcommand{\Weft}{{\cc{W}}}
\newcommand{\W}[1]{{\Weft}{{[#1]}}}

\newcommand{\source}{\ensuremath{{\mathcal S}}}
\newcommand{\target}{\ensuremath{{\mathcal T}}}
\newcommand{\base}{\ensuremath{{\mathcal R}}}

\begin{document}

\title{Computational Short Cuts in Infinite Domain\\
Constraint Satisfaction}
\thanks{This is a preprint of {\em Computational Short Cuts in Infinite Domain Constraint Satisfaction}~\cite{lagerkvist2022e}}
\author{Peter Jonsson}
\address[P. Jonsson]%
   {Dep. Computer and Information
     Science \\Link\"opings  Universitet, Sweden
   }
\email{peter.jonsson@liu.se}   
\author{Victor Lagerkvist}
\address[V. Lagerkvist]%
   {Dep. Computer and Information
     Science \\Link\"opings  Universitet, Sweden}
   \email{victor.lagerkvist@liu.se}

   \author{Sebastian Ordyniak}
\address[S. Ordyniak]%
   {School of Computing \\ University of Leeds, Leeds, UK}
\email{sordyniak@gmail.com}

\maketitle

\begin{abstract}
    A backdoor in a finite-domain CSP instance is a set of variables where each possible
    instantiation moves the instance into a polynomial-time solvable class. Backdoors have found many applications in artificial intelligence and elsewhere, and the algorithmic problem of finding such backdoors has consequently been intensively studied. Sioutis and Janhunen (Proc. 42nd German Conference on AI (KI-2019)) have proposed a generalised backdoor concept suitable for infinite-domain CSP instances over binary constraints. We generalise their concept into a large class of CSPs that allow for higher-arity constraints. We show that this kind of infinite-domain backdoors have many of the positive computational properties that finite-domain backdoors have: the associated computational problems are fixed-parameter tractable whenever the underlying constraint language is finite.  On the other hand, we show that infinite languages make the problems considerably harder: the general backdoor detection problem is \Weft[2]-hard and fixed-parameter tractability is ruled out under standard complexity-theoretic assumptions. We demonstrate that backdoors may have suboptimal behaviour
    on binary constraints---this is detrimental from an AI perspective where binary constraints are predominant in, for instance, spatiotemporal applications.
    In response to this, we introduce {\em sidedoors} as an alternative to backdoors. The fundamental computational problems for
    sidedoors remain fixed-parameter tractable for finite constraint language (possibly also containing non-binary relations). Moreover, the sidedoor approach
    has appealing computational properties that
    sometimes leads to faster algorithms than the backdoor approach.
\end{abstract}  

\section{Introduction}

The {\em constraint satisfaction problem} (CSP) is the widely studied combinatorial problem of determining whether a set of constraints admits at least one solution. It is common to parameterise this problem by a set of relations (a {\em constraint language}) which determines the allowed types of constraints, and by choosing different languages one can model different types of problems. Finite-domain languages e.g.\ makes it possible to formulate Boolean {\em satisfiability} problems and {\em coloring} problems while infinite-domain languages are frequently used to model classical qualitative reasoning problems such as {\em Allen's interval algebra} and the {\em region-connection calculus} (RCC). Under the lens of classical complexity a substantial amount is known: every finite-domain CSP is either tractable or is \NP-complete~\cite{bulatov2017,zhuk20}, and for infinite domains there exists a wealth of dichotomy results separating tractable from intractable cases~\cite{Bodirsky:InfDom}.

The vast expressibility of infinite-domain CSPs makes the search for efficient solution methods extremely worthwhile.
While worst-case complexity results indicate that many interesting problems should be insurmountably hard to solve, they are nevertheless
solved in practice on a regular basis.
The discrepancy between theory and
practice is often explained by the existence of “hidden structure” in real-world
problems~\cite{gaspersS12}. If such a hidden structure exists, then it may be exploited and offer a way of constructing
improved constraint solvers.
To this end, {\em backdoors} have been proposed as a concrete way of exploiting
this structure. 
A backdoor represents a ``short cut'' to solving a hard problem instance and may be seen as a measurement for how close a problem instance is to being polynomial-time solvable~\cite{backdoors2005}. 
%More formally, a backdoor can be defined as a set of variables so that once these have been assigned concrete values, the remaining problem is easy to %solve. 
The existence of a backdoor then allows one to solve a hard problem by brute-force enumeration of assignments to the (hopefully small) backdoor and then solving the resulting problems in polynomial time. This approach has been highly successful: applications can be found in e.g.\
(quantified) propositional satisfiability~\cite{Samer:Szeider:jar2009,Samer:Szeider:HoS}, 
abductive reasoning~\cite{Pfandler:etal:ijcai2013}, argumentation~\cite{Dvorak:etal:ai2012}, planning~\cite{Kronegger:etal:ai2019}, logic~\cite{Meier:etal:algo2019}, and answer set programming~\cite{FICHTE201564}.
Williams et al. (2003) argue that backdoors may explain why SAT solvers occasionally fail to solve randomly generated instances with only a handful of variables but succeed in solving real-world instances containing thousands of variables. This argument appears increasingly relevant since modern SAT solvers frequently handle real-world instances with {\em millions} of variables. Might it be possible to make similar headway for infinite-domain CSP solvers? For example, can solvers in qualitative reasoning (see, e.g., the survey~\cite{Dylla:2017:SQS:3058791.3038927}) be analysed in a backdoor setting? Or are the various problems under consideration so different that a general backdoor definition does not make sense?

\paragraph*{Backdoors for infinite-domain CSPs}
We attack the problem from a general angle and propose a backdoor notion applicable to a large class of infinite-domain CSPs 
with pronounced practical and theoretical interest. 
Our departure is a recent paper~\cite{Sioutis:Janhunen:ki2019} where backdoors are studied for qualitative constraint networks (which corresponds to CSPs over certain restricted sets of binary relations).
We begin by demonstrating why the finite-domain definition of backdoors is inapplicable in the infinite-domain setting
and then continue by presenting our alternative definition, based on the idea of defining a backdoor with respect to {\em relationships} between variables rather than individual variables (which is the basis for the finite-domain definition~\cite{GASPERS201738}). We consider 
CSPs with respect to a fixed set of binary\footnote{The generalisation to higher-arity relations is straightforward.} basic relations $\base$,
that we refer to as the {\em base language}. The base language may, for instance, be
some simple mathematical structure like $({\mathbb Q}; <)$ or a more complex structure
such as
the basic relations in Allen's algebra or RCC-5. 
We 
then consider constraint languages that are definable by
first-order formulas over $\base$. 
For example, {\em temporal constraints} are based on relations that are first-order definable over
$({\mathbb Q}; <)$. The binary relations that are first-order definable
in this structure form the {\em point algebra}~\cite{Vilain:Kautz:aaai86}
and the corresponding CSP is tractable. However,
CSPs based on temporal constraints contain relevant problems that
are not expressible via binary relations; examples include AND/OR precedence constraints
in scheduling~\cite{Mohring:etal:sicomp2004}, the ORD-Horn class in temporal reasoning~\cite{Nebel:Burckert:jacm95}, and ordering CSPs~\cite{Goerdt:csr2009,Guruswami:etal:sicomp2011,Guttmann:Maucher:ifiptcs2006} which subsumes some classical \NP-complete problems such as
\textsc{Betweenness} and \textsc{Cyclic Ordering}~\cite[Problem MS1 and MS2]{gj79}.

Let us now consider two constraint languages $\source$ (the {\em source} language) and $\target$ (the {\em target} language) definable by
first-order formulas over $\base$. 
In this setting we then define a backdoor as a set of pairs of variables so that once the relationship between these variables are fixed, a given CSP$(\source)$
instance is transformed into a CSP$(\target)$: this is advantageous if CSP$(\target)$ is tractable and the backdoor is reasonably small.
The connection between CSP$(\source)$ and CSP$(\target)$ is described via a {\em simplification map} and we show that such maps
can be automatically computed whenever $\source$ and $\target$ are finite languages.
The backdoor approach to a problem now consists of two steps: first a backdoor is computed (the {\em backdoor detection} problem), and then the backdoor is used to solve the
given CSP instance (the {\em backdoor evaluation} problem). 
%Preferably, we want both
%problems to be efficiently solvable. Given that they are typically NP-hard, it is natural
%to use the parameterised complexity framework and analyse these problems
%with respect to backdoor size.
If we contrast our approach with that of Sioutis and Janhunen~\cite{Sioutis:Janhunen:ki2019}, then our method is applicable to CSPs over relations of arbitrarily high arity, and we require only mild, technical assumptions on the set of binary basic relations. Crucially, Sioutis and Janhunen do not consider the computational complexity of any backdoor related problems, and thus do not obtain any algorithmic results.

One of the most important properties of finite-domain backdoors is that they have desirable computational properties.
Unsurprisingly, backdoor detection is \NP-hard under the viewpoint of classical complexity, even for severely restricted cases. However, the situation changes if we adopt a 
{\em parameterized} complexity view. Here, the idea is to approach hard computational problems by characterizing problem parameters that can be expected to be small
in applications, and then design algorithms that are polynomial in the input size combined with a super-polynomial dependence on the parameter.
We say
that a problem is {\em fixed-parameter tractable} (fpt) if its running time is bounded by $f(p) \cdot n^{O(1)}$ where $n$ is the instance size, $p$ the parameter, and
$f$ is a computable function. 
The good news is then that the backdoor detection/evaluation problem for finite-domain CSPs with a fixed finite and tractable constraint language
is fpt when parameterized by the size of the backdoor~\cite{Gaspers:etal:backdoor}.
Thus, if the backdoor size is reasonably small, which we expect for many real-world instances with hidden structure, backdoors can both be found efficiently and be used to simplify the original problem.
However, if the constraint language is not finite, then the basic computational problems may become \Weft[2]-hard~\cite{10.1007/978-3-319-10428-7_18}, and this rules out fpt algorithms under widely accepted
complexity-theoretic assumptions.

While there are profound differences between finite- and infinite-domain CSPs, many important properties of backdoors fortunately remain valid
when switching to the world of infinite domains.
We construct algorithms for backdoor detection and evaluation showing that these problems are fixed-parameter tractable (with respect to the size of the backdoor) for infinite-domain CSPs based on {\em finite} constraint languages. 
Many CSPs studied in practice 
fulfill this condition and our algorithms are directly applicable to such problems. 
%\hl{What's the main differences between our algorithms and previous algorithms for fin-dom backdoors? They are not that %different, but maybe it's at least worth pointing out that they, in contrast to finite domains, are not based on %enumerating domain elements, but rather functions to the base set of relations.}
Algorithms for the corresponding finite-domain problems are based on enumeration of domain values. This is clearly
not possible when handling infinite-domain CSPs, so our algorithms enumerate other kinds of objects, which introduces certain technical difficulties.
Once we leave the safe confinement of finite languages the situation 
changes drastically. We prove that the backdoor detection problem is \W{2}-hard for infinite
languages, making it unlikely to be fixed-parameter tractable. 
Importantly, 
our \W{2}-hardness result is applicable to {\em all} base languages, meaning that it is not possible to circumvent this difficulty by
simply targeting some other base language. Hence, while some cases of hardness are expected, given earlier results for satisfiability 
and finite-domain CSPs~\cite{gaspersS12}, it is perhaps less obvious that essentially all infinite-domain CSPs exhibit
the same source of hardness.

\paragraph*{Sidedoors}

AI research on infinite-domain CSPs has historically focused on binary relations: it is, for example, well-known that
binary relations can capture many relevant facets of space and time, and thus
constitute a suitable basis for spatiotemporal reasoning.
This represents a possible problem for the backdoor approach outlined in the previous section since the possibilities
for substantially faster algorithms appear to be quite limited.
Recall that a backdoor is defined as a set of pairs of variables so that once the relationship between these variables are fixed, the resulting problem belongs to a given tractable class. If the instance we start with only contains
binary relations, then the backdoor approach merely turns into a straightforward method based on replacing
binary constraints with other binary constraints.
This may sound like a highly primitive approach but it is still valuable: it is reassuring to see that it leads
to substantial speed-ups in certain cases (see, for instance, Example~\ref{ex:rcc5}).
However, this observation gives us the impression that one can do better.
To this end, we suggest a new method: the {\em sidedoor} approach.

The idea behind the sidedoor approach is inspired by the behaviour of the backdoor approach on binary constraints. 
Let $I$ be an CSP instance over a language $\source$. Assume that $R(x,y)$
is a binary constraint in $I$ where
relation $R$ is the disjunction of two relations $R_1$ and $R_2$. Then we can do the following simple branching: remove the constraint $R(x,y)$ from $I$
and replace it with $R_1(x,y)$ in the first branch and with $R_2(x,y)$ in the second, and thereafter solve the resulting instances 
recursively. 
If the branches of this process ends in instances over some language $\target$, then we can solve $I$ with the aid of a solver
for CSP$(\target)$.

We generalise this idea by branching on
subinstances containing a certain number of variables (which we call the {\em radius}) and replacing the constraints within the subinstances with other constraints
in a way that preserves the solutions. This generalisation allows the algorithm to cover several constraints in each branch. We see, for instance, that the number of variable pairs
that are covered by a sidedoor increases quadratically with its radius. 
Hence, the number of binary constraints
that can be covered at once increases quite rapidly with increasing radius.
If this procedure leads to instances within a tractable class and the number of subinstances
that we need to replace is sufficiently small, then it appears
to be a viable approach. One should also note that this approach is not geared towards
infinite domains and it works for general CSPs regardless of domain size.

Just like the backdoor approach, we have two fundamental computational problems: {\em sidedoor detection} and {\em sidedoor evaluation}.
We prove that
the detection and evaluation problems for sidedoors are fixed-parameter tractable (when parameterized by the size of the sidedoor) under mild additional assumptions. 
The detection problem can be solved in $O((rs)^{rs})$ time where $s$ is sidedoor size
and $r$ is the chosen radius---thus, the detection problem is fpt even for infinite constraint languages.
However, there is an important caveat: the radius must satisfy $r \geq a$ where $a$ is the maximum arity of the relations
in the given instance. Thus, even if the detection problem is formally fpt for fixed source and target languages together with a fixed radius, 
the method is restricted to a subset of the
possible instances. 
%We note, however, that there are infinite languages (with infinite domains) of bounded arity. One concrete example is
%$\Gamma=\{\dots,R_{-2},R_{-1},0,R_1,R_2,\dots\}$ where $R_k=\{(x,y) \in {\mathbb Z}^2 \; | \; y=x+k\}$.

The main obstacle for applying the sidedoor approach is to identify and describe suitable branchings on subinstances: we do this via
{\em branching maps}. Unlike simplification maps for backdoors, the computation of branching maps
appears to be a difficult problem. Another important difference is that the exact choice of branching map 
strongly affects the complexity of sidedoor evaluation; this is not the case for the backdoor
evaluation problem. We present an
algorithmic method to compute branching maps for languages that satisfy a particular definability condition.
We stress that this method does not produce branching maps with optimal behaviour and that more research is needed into the construction of branching maps.

We conclude with a remark: the sidedoor approach exploits another kind of hidden structure than the backdoor approach. While backdoors
focus on the tuples within the relations, sidedoors focus on how constraints are connected to each other.
Thus, backdoors and sidedoors may be viewed as complementary methods based on the kind of structure they exploit.
This distinction resembles the two main paths that have been followed in the search for tractable CSP fragments~\cite{Carbonnel:Cooper:constraints2016}: either one studies the CSP with a restricted set of allowed relations, or one
studies the CSP where restrictions are imposed on which variables may be mutually constrained.

\paragraph*{Outline}

The article has the following structure.
Section~\ref{sec:prel} presents some technical background concerning constraint satisfaction and parameterised complexity.
We present the generalisation of backdoors to infinite domains in
Section~\ref{sec:fin}
and we introduce sidedoors in Section~\ref{sec:sidedoors}.
Section~\ref{sec:concludingremarks} concludes the article
with a summary and a comprehensive discussion concerning future research directions.  
This article is an extended version of a conference paper~\cite{Jonsson:etal:cp2021} and the main extension is the addition of the sidedoor concept. The reader should be aware of the fact that we have changed some of the underlying assumptions and the terminology in this article
compared to the conference paper. 
For instance, we now assume that the base structure in the backdoor approach is a partition scheme (in the sense of Ligozat and Renz (2004)) in order to streamline
the presentation. Concerning the terminology, we have for example
changed {\em language triple}
into {\em backdoor triple} to emphasise the connection with backdoors.

\section{Preliminaries}
\label{sec:prel}

This section presents the basic definitions that we need in this article.

\subsection{Relations and Formulas}

A {\em relational signature} $\tau$ is a set of relation symbols $R_i$ where each $R_i$ is associated with an {\em arity} $k \in {\mathbb N}$.  A {\em relational structure} $\Gamma$ over the relational signature $\tau$ is a set $D_{\Gamma}$ (the {\em domain}) together with a relation
$R_i^{\Gamma} \subseteq D_{\Gamma}^{k_i}$ for each relation symbol $R_i$ in $\tau$ of arity $k_i$.
For simplicity, we often (1)
do not distinguish between the signature of a relational structure and its relations, and
(2) view a relational structure as a set of relations.

Assume that the relational structure $\base$ is a set of binary relations over the domain $D$. We say 
that the relations in $\base$ are {\em jointly exhaustive} (JE)
if $\bigcup \base = D^2$, and that they are
{\em pairwise disjoint} (PD)
if $R \cap R' = \varnothing$ for all distinct $R,R' \in \base$. Additionally, a relational structure which is both JE and PD is said to be JEPD.
We say that $\base$ is a {\em partition scheme}~\cite{Ligozat:Renz:pricai2004} if it is JEPD, it contains the equality relation
$\{(d,d) \; | \; d \in D\}$, and for every $R \in \base$, the language $\base$ contains the converse relation $R^{\smile}=\{(d',d) \; | \; (d,d') \in R\}$.

We continue by recalling some terminology from logic.  First-order formulas over a signature $\tau$ 
are inductively defined using the logical symbols of universal and existential quantification, disjunction, conjunction, negation, equality, variable symbols and the relation symbols in $\tau$. A $\tau$-formula without free variables is called a $\tau$-{\em sentence}.  
We write $\Gamma \models \phi$ if the relational structure $\Gamma$ (with
signature $\tau$) is a model for the $\tau$-sentence $\phi$.

Let $\varphi(x_1, \ldots, x_n)$ be a first-order formula (with equality) over free variables $x_1, \ldots, x_n$  over a relational structure $\Gamma=(D; R_1, \ldots, R_m)$.  We write $\sols(\varphi(x_1, \ldots, x_n))$ for the set of models of $\varphi(x_1, \ldots, x_n)$ with respect to $x_1, \ldots, x_n$, i.e., \[(d_1, \ldots, d_n) \in \sols(\varphi(x_1, \ldots, x_n))\]
if and only if 
\[(D; R_1, \ldots, R_m) \models \varphi(d_1, \ldots, d_n),\] and we write
\[R(x_1, \ldots, x_n) \equiv \varphi(x_1, \ldots, x_n)\] to indicate that
the relation $R$ equals $\sols(\varphi(x_1, \ldots, x_n))$.
In this case, we say that $R$ is {\em first-order definable} (fo-definable) in $\Gamma$.
Note that our definitions are always parameter-free, i.e. we do not allow the use of domain elements in them.
In addition to first-order logic, we sometimes use the quantifier-free (qffo), the primitive positive (pp), and
the quantifier-free primitive positive (qfpp) fragments. 
The qffo fragment consists of all formulas without quantifiers, the pp fragment
consists of formulas that are built using existential quantifiers, conjunction and equality, and the
qfpp consists of quantifier-free pp formulas. We lift the notion of definability to these fragments in the
obvious way. 

If the structure $\Gamma$ admits quantifier elimination (i.e. every first-order formula has a logically
equivalent formula without quantifiers), then fo-definability coincides with qffo-definability. 
This is sometimes relevant in the sequel since our results are mostly based on qffo-definability. However, it is important to note that if $R$ is pp-definable in $\Gamma$, then $R$ is not necessarily
qfpp-definable in $\Gamma$ even if $\Gamma$ admits quantifier elimination.
There is
a large number of structures admitting quantifier elimination and interesting examples are presented
in every standard textbook on model theory, cf. Hodges (1993)\nocite{HodgesLong}.
Well-known examples include {\em Allen's interval algebra} (under the standard representation via intervals in ${\mathbb Q}$)
and the spatial formalisms RCC-5 and RCC-8 (under the model-theoretically pleasant representation suggested by Bodirsky \& Wölfl (2011)\nocite{Bodirsky:Wolfl:ijcai2011}).
Some general quantifier elimination results that are
highly relevant for computer science and AI are discussed in Bodirsky~\cite[Sec. 4.3.1]{Bodirsky:InfDom}.

%In addition, if $\varphi(x_1, \ldots, x_n)$
%does not contain any quantifiers, then $R$ is {\em quantifier-free} fo-definable (qffo-definable), and
 %   if $\varphi(x_1, \ldots, x_n)$ consists only of existential quantifiers and conjunctions of positive literals, then $R$ is
  %  {\em primitive positive definable} (pp-definable).
%Finally, if the formula is quantifier-free then $R$ is quantifier-free pp-definable (qfpp-definable). 

\subsection{The Constraint Satisfaction Problem} \label{sec:csp}
Let $\Gamma_D$ denote a relational structure with domain $D$. 
  The
{\em constraint satisfaction problem} over
$\Gamma_D$ (CSP$(\Gamma_D)$) is the computational problem of
determining whether a set of constraints over $\Gamma_D$ admits at
least one satisfying assignment. 

\pbDef{CSP$(\Gamma_D)$}
      {A tuple $(V,C)$ where $V$ is a set of variables and $C$ a set of constraints of the form $R(x_1, \ldots, x_k)$, where $R \in \Gamma_D$ and $x_1, \ldots, x_k \in V$.} {Does there exists a satisfying assignment to $(V,C)$, i.e., a function $f \colon V \rightarrow D$ such that $(f(x_1), \ldots, f(x_k)) \in R$ for each constraint $R(x_1, \ldots, x_K) \in C$?}

The relational structure $\Gamma_D$ is often called a {\em constraint language}.
We slightly abuse notation and write $\sols(I)$ for the set of all satisfying assignments to a CSP$(\Gamma)$ instance $I$. 
Finite-domain constraints admit a simple
representation obtained by explicitly listing all tuples in the
involved relation. 
For infinite domain CSPs, it is frequently assumed that $\Gamma$ is a {\em first-order reduct} of an underlying relational structure ${\mathcal R}$, i.e.,  each $R \in \Gamma$ is 
fo-definable in ${\mathcal R}$. Whenever ${\mathcal R}$ admits quantifier elimination, then we can always work with the
{\em qffo reduct} where each $R \in \Gamma$ is qffo-definable in ${\mathcal R}$. 
%an underlying relational
%structure, a {\em first-order reduct}.

\begin{example} \label{ex:eq}
  An {\em equality} language is a first-order reduct of a structure $(D; \emptyset)$ where $D$ is a countably
  infinite domain. Each literal in a first-order formula over this structure is either of the form $x = y$, or 
  $x \neq y \equiv \neg (x =
  y)$. The structure  $(D; \emptyset)$ admits quantifier elimination so
  every first-order reduct can be viewed as a qffo reduct. For example, if we let relation $\delta \subseteq D^3$ be defined via the formula
  $(x = y \land x \neq z) \lor (x \neq y \land y = z)$, then CSP$(\{\delta\})$ is known to be \NP-complete~\cite{bodirskykara2010}. On the other hand, CSP$(\{=, \neq\})$ is well-known to be tractable. 
    
    A {\em
    temporal language} is a first-order reduct of
 $(\mathbb{Q}; <)$. The structure  $(\mathbb{Q}; <)$ admits quantifier elimination so
  it is sufficient to consider qffo reducts.  For example, the {\em betweenness relation}
  $\beta=\{(x,y,z) \in {\mathbb Q}^3 \; | \; \min(x,z) < y < \max(x,z)\}$
 can be defined via the formula $(x<y \land y < z)
 \lor (z<y \land y <x)$, and the resulting CSP is well-known to be \NP-complete.
 
 It will often be useful to assume that the underlying relational structure is JEPD or that it is a partition scheme. Clearly, neither $(\mathbb{N}; \emptyset)$, nor $(\mathbb{Q}; <)$ are JEPD, but they can easily be expanded to satisfy the JEPD condition by (1) adding the converse of each relation, and (2) adding the complement of each relation. Thus, an equality language can be defined as a first-order reduct of the partition scheme $(\mathbb{N}; =, \neq)$, and a temporal language as a first-order reduct of the partition scheme $(\mathbb{Q}; =, <, >)$. More ideas for
 transforming non-JEPD languages into JEPD languages and partition schemes can be found in Bodirsky \& Jonsson~\cite[Sec. 4.2]{DBLP:journals/jair/BodirskyJ17}.
 
\end{example}  

 %\todo{Do we need to introduce temporal languages at all? It depends on %whether we can use temporal languages to contrast equality languages.}

  Constraint languages in this framework capture many problems of
  particular interest in artificial intelligence. For example,
  consider the {\em region connection calculus} with the 5 basic relations
  $\Theta=\{{\sf DR},{\sf PO},{\sf PP},{\sf PP}^{-1},{\sf
    EQ}\}$ (RCC-5). 
    See Figure~\ref{fig:rcc5} for a visualisation of
  these relations. These five relations obviously form a partition scheme.
  In the traditional formulation of this calculus one then
  allows unions of the basic relations, which (for
  two regions $X$ and $Y$) e.g.\ allows us to
  express that $X$ is a proper part of $Y$ or $X$ and $Y$ are equal.
  This relation can easily be defined
  via the (quantifier-free) first-order formula $(x {\sf PP} y) \lor (x {\sf EQ}
  y)$. Hence, if we let $\Theta^{\vee =}$ be the constraint language
  consisting of all unions of basic relations in $\Theta$, then
$\Theta^{\vee =}$ is a qffo reduct of $\Theta$.

\begin{figure}
    \centering
    \captionsetup[subfigure]{labelformat=empty}
    \begin{subfigure}[b]{0.19\textwidth}
        \centering
        \begin{tikzpicture}[scale=0.65]
            \draw (0,0) circle (1);
            \draw[dashed] (0,0) circle (1);
            \draw (-0.33,0) node {$X$};
            \draw (0.33,0) node {$Y$};
        \end{tikzpicture}
        \caption{${\sf EQ}(X,Y)$}
    \end{subfigure}
    \hfill
    \begin{subfigure}[b]{0.19\textwidth}
        \centering
        \begin{tikzpicture}[scale=0.4]
            \draw (0,0.75) circle (1);
            \draw[dashed] (0,-0.75) circle (1);
            \draw (0,0.75) node {$X$};
            \draw (0,-0.75) node {$Y$};
        \end{tikzpicture}
        \caption{${\sf PO}(X,Y)$}
    \end{subfigure}
    \hfill
    \begin{subfigure}[b]{0.19\textwidth}
        \centering
        \begin{tikzpicture}
            \draw (-0.3,0) circle (0.4);
            \draw (-0.3,0) node {$X$};
            \draw[dashed] (0,0) circle (1);
            \draw (0.5,0) node {$Y$};
        \end{tikzpicture}
        \caption{${\sf PP}(X,Y)$}
    \end{subfigure}
    \hfill
    \begin{subfigure}[b]{0.19\textwidth}
        \centering
        \begin{tikzpicture}
            \draw[dashed] (-0.3,0) circle (0.4);
            \draw (-0.3,0) node {$Y$};
            \draw (0,0) circle (1);
            \draw (0.5,0) node {$X$};
        \end{tikzpicture}
        \caption{${\sf PP}^{-1}(X,Y)$}
    \end{subfigure}
    \hfill
    \begin{subfigure}[b]{0.19\textwidth}
        \centering
        \begin{tikzpicture}[scale=0.4]
            \draw (0,1.1) circle (1);
            \draw (0,1.1) node {$X$};
            \draw[dashed] (0,-1.1) circle (1);
            \draw (0,-1.1) node {$Y$};
        \end{tikzpicture}
        \caption{${\sf DR}(X,Y)$}
    \end{subfigure}
    
    \caption{Illustration of the basic relations of RCC-5 with 
    two-dimensional disks.}
    \label{fig:rcc5}
\end{figure}
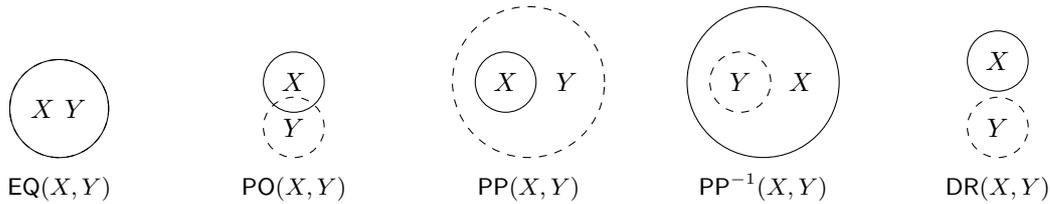

\subsection{Parameterized Complexity}
 
To analyse the complexity of CSPs, we use the framework of 
{\em parameterized complexity}~\cite{rod1999parameterized,book/FlumG06}
where the run-time of an algorithm is studied with respect to a parameter
$p \in {\mathbb N}$ and the input size~$n$.
Given an instance $I$ of some
computational problem, we let $||I||$ denote the bit-size of $I$.
Many important CSPs are \NP-hard on general instances and are regarded
as being intractable. However, realistic problem instances are not chosen arbitrarily and they often
contain structure that can be exploited for solving the instance efficiently.
The idea behind parameterized analysis is that the parameter describes the structure of
the instance in a computationally meaningful way.
The result is a fine-grained complexity analysis
that is more relevant to real-world problems while still admitting a rigorous theoretical treatment including, for instance,
algorithmic performance guarantees.
The most favourable complexity class is \FPT\
(\emph{fixed-parameter tractable})
which contains all problems that can be decided 
in $f(p)\cdot n^{O(1)}$ time, where $f$ is a computable
function.

%The next best option is the
%complexity class \XP, which contains all problems decidable
%in $n^{f(p)}$ time, i.e. the problems solvable in polynomial time
%when the parameter $p$ is bounded.
%Clearly, $\FPT \subseteq \XP$ and
%this inclusion is strict 
%(see e.g.~\cite[Corr. 2.26]{book/FlumG06}).
%It is significantly 
%better if a problem is in \FPT than in \XP,
%since the order of the polynomial factor in the former case does not depend on the parameter $p$.
%However, parameterised complexity offers strong theoretical evidence that
%inclusion in \FPT is highly unlikely for some problem, e.g. 
%those that are hard for the class \Weft[1].
%The latter contains all problems that admit many-to-one reductions
%from the {\sc Parameterised Clique} problem, which asks whether a graph
%has a clique of size $p$, where $p$ is the parameter,
%and the reduction runs in fixed-parameter time with respect to $p$.
%Finally, a problem is \paraNP-hard if it is \NP-hard for some constant value of the parameter. 
%A problem that is \paraNP-hard cannot be in \XP unless \P = \NP.

We will prove that certain problems are not in $\FPT$
and this requires some machinery.
A {\em parameterized problem} is, formally speaking, a subset of $\Sigma^* \times {\mathbb N}$
where $\Sigma$ is the input alphabet. Reductions between parameterized problems need to take
the parameter into account. To this end, we will use {\em parameterized reductions} (or fpt-reductions).
Let $L_1$ and $L_2$ denote parameterized problems with $L_1 \subseteq \Sigma_1^* \times {\mathbb N}$
and $L_2 \subseteq \Sigma_2^* \times {\mathbb N}$. 
A parameterized reduction from $L_1$ to $L_2$ is a
mapping $P: \Sigma_1^* \times {\mathbb N} \rightarrow \Sigma_2^* \times {\mathbb N}$
such that
(1) $(x, k) \in  L_1$ if and only if $P((x, k)) \in L_2$, (2) the mapping can be computed
by an fpt-algorithm with respect to the parameter $k$, and (3) there is a computable function $g : {\mathbb N} \rightarrow {\mathbb N}$ 
such that for all $(x,k) \in L_1$ if $(x', k') = P((x, k))$, then $k' \leq g(k)$.
The class $\Weft[1]$ contains all problems that are fpt-reducible to Independent Set when parameterized
by the size of the solution, i.e. the number of vertices in the independent set.
Showing $\Weft[1]$-hardness (by an fpt-reduction) for a problem rules out the existence of a fixed-parameter
algorithm under the standard assumption $\FPT \neq \Weft[1]$.
$\Weft[1]$ is a complexity class in the {\em weft hierarchy}
$\Weft[1] \subseteq \Weft[2] \subseteq \dots$. We note that
proving $\Weft[k]$-hardness, $k \geq 1$, for a problem
implies $\Weft[1]$-hardness and the nonexistence of fixed-parameter algorithms.

\section{Backdoors}
\label{sec:fin}

This section is devoted to the introduction of a general backdoor concept
for CSPs. The underlying motivation is discussed in Section~\ref{sec:backdoor-motivation}
and the formal definition of backdoors is presented in~\ref{sec:basic-defs}.
Section~\ref{sec:algs} contains algorithms for various backdoor problems over finite constraint languages.
Finally, certain hardness results are proved in Section~\ref{sec:inf} for backdoor problems over infinite constraint languages.

\subsection{Motivation}
\label{sec:backdoor-motivation}

We begin by recapitulating the standard definition of backdoors for finite-domain CSPs.
Let $\alpha \colon X \rightarrow D$ be an assignment. For a $k$-ary
constraint $c = R(x_1, \ldots, x_k)$ 
we denote by $c_{\mid \alpha}$ the constraint over the relation $R_0$ and with scope $X_0$ obtained from $c$ as follows: $R_0$ is obtained from
$R$ by 

\begin{enumerate}
\item
removing $(d_1 ,\ldots, d_k)$ from $R$ if there exists $1 \leq i \leq k$ such that
$x_i \in X$ and $\alpha(x_i) \neq d_i$, and 

\item
removing from all remaining tuples all coordinates $d_i$ with
$x_i \in X$.
\end{enumerate}

The scope $X_0$ is obtained from $x_1, \ldots, x_k$ by removing every $x_i \in X$. For a set $C$ of
constraints we define $C_{\mid \alpha}$ as $\{c_{\mid \alpha} \colon c \in C\}$. We now have everything in place to define the standard notion of a (strong) backdoor, in the context of Boolean satisfiability problems and finite-domain CSPs.

\begin{definition}[See, for instance, \cite{GASPERS201738} or \cite{williams2003}] \label{def:fin_backdoor}
Let ${\mathcal H}$ be a set of CSP instances. A {\em ${\mathcal H}$-backdoor}
for a CSP$(\Gamma_D)$ instance $(V,C)$ is a set $B \subseteq V$ where
$(V \setminus B, C_{\mid \alpha}) \in {\mathcal H}$ for each $\alpha \colon
B \rightarrow D$. 
\end{definition}

\iflong
In practice, ${\mathcal H}$ is typically defined as a
polynomial-time solvable subclass of CSP and one is thus interested in finding a backdoor into the tractable class ${\mathcal H}$.
If the CSP instance $I$ has a backdoor of size $k$, then it can be solved in $|D|^k \cdot {\rm poly}(||I||)$ time.
This is an exponential running time with the advantageous feature that it is
exponential not in the instance size $||I||$, but in the domain size and backdoor set size only.

\fi
%As we will see in just a moment, Definition~\ref{def:fin_backdoor} turns out to be troublesome in the current setting of infinite domains.

%Morover, we
%assume that $(D; R_1, \ldots, R_m)$ admits quantifier-elimination and
%assume that each relation in $\Gamma$ is defined as a quantifier-free
%formula in conjunctive normal form over $\{R_1, \ldots, R_m\}$\footnote{It might be easier to follow the forthcoming definitions by having the examples of equality languages or temporal languages in mind.}.

\begin{example}
  Let us first see why Definition~\ref{def:fin_backdoor} is less
  impactful for infinite-domain CSPs. Naturally, the most obvious problem is that one, even for a fixed $B \subseteq V$, need to consider infinitely many functions $\alpha \colon V \rightarrow D$, and there is thus no general argument which resolves the backdoor evaluation problem.
  However, even for a fixed assignment $\alpha \colon V \rightarrow D$ we may run into severe problems. Consider a single equality
  constraint of the form $(x = y)$ and an assignment $\alpha$ where
  $\alpha(x) = 0$ but where $\alpha$ is not defined on $y$. Then $(x =
  y)_{\mid \alpha} = \{(0)\}$, i.e., the constant $0$ relation, 
  which is {\em not} an equality relation.
  %Hence, there
  %does not exist {\em any} tractable equality language which could
  %serve as a ... 
  Similarly, consider a constraint $X r Y$ where $r$ is a basic
  relation in RCC-5. Regardless of $r$, assigning a fixed
  region to $X$ but not to $Y$ results in a CSP instance which is not
  included in {\em any} tractable subclass of RCC-5 (and is not even an RCC-5 instance). 
\end{example}  

Hence, the usual definition of a backdoor fails to compensate for a
fundamental difference between finite and infinite-domain CSPs: that
assignments to variables are typically much less important
than the {\em relation} between variables.  
%This is reflected by the following fact: a CSPs with a finite domain $\{1,\dots,d\}$ 
%can {\em always}
%be viewed as reduct of the unary structure $(\{1, \ldots, d\}; \{(1)\}, \ldots, \{(d)\})$,
%while CSPs for infinite domains are typically defined as reducts of structures
%that are not unary.

\subsection{Basic Definitions and Examples}
\label{sec:basic-defs}

Recall that we in the infinite-domain setting are mainly interested in CSP$(\Gamma)$ problems where each relation in $\Gamma$ is qffo-definable over a fixed relational structure $\base$. Hence, in the backdoor setting we obtain three components: a relational structure $\base$ and two qffo reducts $\source$ and $\target$ over $\base$, where CSP$(\source)$ is the (likely \NP-hard) problem which we want to solve by finding a backdoor to the (likely tractable) problem CSP$(\target)$. Additionally, we will assume that $\base$ is a partition scheme.

%, a target language $\target$, and a source language $\source$

\begin{definition}
Let $\base=(D; R_1, R_2, \dots)$ be a partition scheme and let $\source,\target$ be
qffo reducts of $\base$. We say that $[\source,\target,\base]$ is a {\em backdoor triple} and we refer to
\begin{itemize}
\item
$\source$ as the {\em source language},

\item
$\target$ as the {\em target language}, and

\item
$\base$ as the {\em base language}.
\end{itemize}
\end{definition}

%Given a language triple $[\source,\target,\base]$, the problem CSP$(\source)$ is the (likely intractable) problem that we want
%to solve, and where CSP$(\target)$ is the (likely tractable) problem
%which we want to reduce to with a backdoor approach. 
Note that $\source$ and $\target$ may contain non-binary relations even though $\base$ only contains binary relations. One should note that all concepts work equivalently well for higher arity relations in $\base$ 
but it complicates the presentation. Generalisations of partition schemes to higher-arity relations
is, for instance, described and discussed by Dylla et al.~\cite[Section 3.1]{Dylla:2017:SQS:3058791.3038927}.

%Also note that the equality relation needs to be qffo-definable in $\base$ since this relation is always available
%in first-order formulas. An alternative way is to require that the equality relation is a member of $\base$ but this
%assumption is stronger than is needed for our purposes (see Example~\ref{ex:fin_dom}).

We continue by describing how constraints can be simplified in the presence of a partial assignment of relations from $\base$ to
pairs of variables.

\iflong
\begin{definition}
  Let $[\source, \target, \base]$ be a backdoor triple and let $(V,C)$ be an instance of CSP$(\source)$.
Say that a partial mapping $\alpha \colon B \rightarrow \base$ for $B \subseteq V^2$ is \emph{consistent} if the CSP$(\base)$ instance \[(V,\{R(x,y) \; | \; \mbox{$x,y \in V$, $\alpha(x,y)$ is defined, $R=\alpha(x,y)$}\})\] is satisfiable.
%We say that $\alpha$ is \emph{consistent} if the CSP$(\source)$ instance with variables $\bigcup_{b\in B}b$ having one constraint with scope $b$ and %relation $\alpha(b)$ for every $b \in B$ is satisfiable.
We define a {\em reduced constraint} with respect to a consistent $\alpha$ as:
%\begin{enumerate}
%\item
 \[R(x_1, \ldots, x_k)_{\mid \alpha} = R(x_1, \ldots, x_k) \land \bigwedge_{\alpha(x_i, x_j) = S, x_i, x_j \in \{x_1, \ldots, x_k\}}S(x_i, x_j).\]
%\item
%  $(V,C)_{\mid \alpha} = (V, \{c_{\mid \alpha} \mid c \in C\})$.
%\end{enumerate}  
\end{definition}
\fi
Next, we describe how reduced constraints can be translated to the target language. Let $[\source, \target, \base]$ be a backdoor triple and let $a \in {\mathbb N} \cup \{\infty\}$ 
equal $\sup \{i \; | \; \mbox{$R \in \source$ has arity $i$}\}$. 
%be the highest arity of any relation in $\source$. 
Let 
\[{\bf S} = \{\sols(R(x_1, \ldots, x_k)_{\mid \alpha}) \mid R \in \source, \; \alpha \colon \{x_1, \ldots, x_k\}^2 \rightarrow \base\}\] and 
\[{\bf T} = \{\varphi(x_1, \ldots, x_k) \mid k \leq a, \varphi(x_1, \ldots, x_k) \text{ is a qfpp-definition over } \target\}.\] 
We interpret these two sets as follows. Each mapping $\alpha \colon \{x_1, \ldots, x_k\}^2 \rightarrow {\mathcal \base}$ applied to a constraint $R(x_1, \ldots, x_k)$ results in a (potentially) simplified constraint $R(x_1, \ldots, x_k)_{\mid \alpha}$, which might or might not be expressible via a CSP$(\target)$ instance. Then the condition that $\sols(R(x_1, \ldots, x_k)_{\mid \alpha}) \in {\bf S}$ is qfpp-definable over $\source$ simply means that the set of models of the constraint $R(x_1, \ldots, x_k)_{\mid \alpha}$ can be defined as the set of models of a CSP$(\source)$ instance. Thus, the set ${\bf S}$ represents all possible simplifications of constraints (with respect to $\source$) and the set ${\bf T}$ represents all possibilities of expressing constraints (up to a fixed arity) by the language $\target$. Crucially, note that  ${\bf S}$ and ${\bf T}$ are finite whenever $\source$ and $\target$ are finite. With the help of the two sets ${\bf S}$ and ${\bf T}$ we then define the following method for translating (simplified) $\source$-constraints into $\target$-constraints.

\begin{definition} \label{def:simp-map}
A {\em simplification map} is a partial mapping $\Sigma$ from ${\bf S}$ to ${\bf T}$ such that for every $R \in {\bf S}$: $\Sigma(R) = \varphi(x_1, \ldots, x_k)$ if $R(x_1, \ldots, x_k) \equiv \varphi(x_1, \ldots, x_k)$ for some $\varphi(x_1, \ldots, x_k) \in {\bf T}$, and is undefined otherwise.
\end{definition}

We typically say that a simplification map goes from the source language $\source$ to the target language $\target$ even though it technically speaking is a map from ${\bf S}$ to ${\bf T}$.
% \iflong
% The above definition is defined via logical formulas (qfpp-definitions instead of CSP$(\target)$ instances. But we could very easily let ${\bf T}$ be the set of CSP$(\target)$ instances with at most $a$ variables and change the condition of a simplification map to: $\Sigma(R) = I$ if $I \in {\bf T}$ and $\sols(I) = R$.
% \fi
%If a simplification map between two languages has the minimum number of undefined entries then we say that it is {\em optimal}. 
 Note that if $\source$ and $\target$ are both finite, then there always exists a simplification map of finite size, and one may
 without loss of generality assume that it is possible to access the map in constant time. We will take a closer look at the computation of simplification maps for finite language in Section~\ref{sec:simp-map}. If $\source$ is infinite, then
 the situation changes significantly. First of all, a simplification map has an infinite number of inputs
 and we cannot assume that it is possible to access it in constant time. We need, however, always assume that it can be
 accessed in polynomial time. Another problem is that we a priori have no general way of computing simplification maps
 so they need to be constructed on a case-by-case basis; we return to this problem in Section~\ref{sec:simp-map}.

\begin{definition} \label{def:backdoor}
  Let $[\source, \target, \base]$ be a backdoor triple, and let $\Sigma$ be a simplification map from $\source$ to $\target$.
  For an instance $(V,C)$ of CSP$(\source)$ we say that $B \subseteq V^2$ is a {\em backdoor} if, for every consistent $\alpha \colon B \rightarrow \base$,
  $\Sigma(R(x_1, \ldots, x_k)_{\mid \alpha})$ is defined for every constraint $R(x_1, \ldots, x_k) \in C$.
\end{definition}  

%If CSP$(\target)$ is polynomial-time solvable and 
%the CSP$(\source)$ instance $I$ has a backdoor of size $k$, then it can be solved in $|\base|^k \cdot {\rm poly}(||I||)$ time.
%This gives us a situation that is similar to the finite-domain case: we have an exponential time algorithm
%but it is only exponential in $|\base|$ and the backdoor set size (and not in the size of the instance). \todo{Do we really want to say %this here? It clashes a bit with Theorem 15 (backdoor evaluation).}

Let us consider a few examples before turning to computational aspects of finding and using backdoors.
We use equality languages as a first illustration of backdoors into infinite-domain CSPs.

%
%\todo{Note: if we require equality to be a member of the JEPD structure $\base$, then the unary relation $(d)$ cannot be %qffo-defined in
%$\base$. However, it can be fo-defined: $x=d \Leftrightarrow \exists y.R_{dd'}(x,y)$ where $d'$ is chosen such that
%$d \neq d'$.}

\begin{example} \label{ex:languages}
Consider equality languages, i.e. languages that are fo-definable over the base structure $(\mathbb{N}; =, \neq)$.
Recall the ternary relation $\delta$ from Example~\ref{ex:eq} and consider a simplification map $\Sigma$ with respect to the tractable target language $\{=,\neq\}$. Note that we {\em cannot} simplify an arbitrary constraint $\delta(x,y,z)$, but that we can simplify $\delta(x,y,z)_{\mid \alpha}$ if (e.g.) $\alpha(x,y)$ is '=', or if $\alpha(x,y)$ is $\neq$. 
 Hence, it is important to realise that we cannot simplify all constraints, but if we end up in a situation where the reduced instance cannot be simplified into $\{=,\neq\}$, then this simply means that the underlying set $B \subseteq V^2$ was not chosen correctly.
  Let $(V,C)$ be an instance of the \NP-hard problem CSP$(\{\delta\})$. Consider the set
  $B = \{(x,y) \mid \delta(x,y,z) \in C\} \subseteq V^2$. We claim that $B$ is a backdoor with respect to $\{=, \neq\}$. Let $\alpha \colon B \rightarrow \{=,\neq\}$, and consider an arbitrary constraint $\delta(x,y,z) \in C$. Clearly, $(x,y) \in B$. Then, regardless of the relation between $x$ and $y$, the constraint can be removed and replaced by $\{=,\neq\}$-constraints.
\end{example}

%\begin{example}
% Let $\Gamma$ be a language over a finite domain $\{1, \ldots, d\}$. Note that such a $\Gamma$ may be viewed as a first-order reduct of the unary structure $(\{1, \ldots, d\}; \{(1)\}, \ldots, %\{(d)\})$. In this setting each backdoor in the style of Definition~\ref{def:fin_backdoor}, defined as a set of variables, is then just a special case of a backdoor in the style of %Definition~\ref{def:backdoor}, defined as a set of pairs of variables. Hence, our notion of a backdoor is not merely an adaption of Definition~\ref{def:fin_backdoor} into infinite domains, but a strict %generalisation.
% \end{example}

The next example is particularly important and it will be revisited in Section~\ref{sec:motex} and in Examples~\ref{ex:sidedoor-1}
and \ref{ex:sidedoor-2}.

\begin{example} \label{ex:rcc5}
Recall the definitions of $\Theta$ and $\Theta^{\vee =}$ for RCC-5 from Section~\ref{sec:prel}.
The problem CSP$(\Theta)$ is known to be polynomial-time solvable
while CSP$(\Theta^{\vee =})$ is \NP-complete~\cite{DBLP:journals/ai/RenzN99}.
Consider a reduced constraint $R_{\mid \alpha}(x,y)$ with respect to
an instance $(V,C)$ of CSP$(\Theta^{\vee =})$, a set $B \subseteq
V^2$, and a function $\alpha \colon B \rightarrow \Theta$. If $(x,y) \in B$ (or, symmetrically, $(y,x) \in B$) then
$R(x,y) \land (\alpha(x,y))(x,y)$ is (1) unsatisfiable if
$\alpha(x,y) \cap R = \emptyset$, or (2) equivalent to
$\alpha(x,y)$. Hence, the simplification map in this case either
outputs an unsatisfiable CSP$(\Theta)$ instance or replaces the
constraints with the equivalent constraint over a basic relation. This
results in an $O(5^{|B|})\cdot \mathrm{poly}(||I||)$ time algorithm
for RCC-5, which can slightly be improved to $O(4^{|B|})\cdot
\mathrm{poly}(||I||)$ with the observation that only the trivial
relation $({\sf DR} \cup {\sf PO} \cup {\sf PP} \cup {\sf PP}^{-1}
\cup {\sf EQ})$ contains all the five basic relations.
The currently fastest known algorithm for CSP$(\Theta^{\vee =})$ runs in
$2^{O(n \log n)}$ time where $n$ is the number of variables~\cite{Jonsson:etal:ai2021}.
Hence, backdoors provide an appealing approach for instances having
small backdoors: the backdoor approach is superior for instances 
with backdoor size in $o(n \log n)$.
\end{example}

The languages considered in Example~\ref{ex:languages} and Example~\ref{ex:rcc5} are special cases of so-called {\em homogeneous} and {\em finitely bounded} languages, where a complexity dichotomy has been conjectured (see, e.g.,~\cite{barto2016}). We refrain from stating this conjecture and the aforementioned properties formally, but remark that such languages form natural examples of languages triples where our backdoor approach is applicable. In fact, the following holds:

\begin{example} \label{ex:finbound}
  If $\base$ is a finitely bounded homogeneous structure, then $[\source,\target,\base]$ is a backdoor triple for {\em any} first-order reducts $\source,\target$ of $\base$. 
\end{example}

Last, we illustrate that the finite domain case can be seen a special case of our backdoor notion.

\begin{example} \label{ex:fin_dom}
Let $D= \{1, \ldots, d\}$ for some $d \in {\mathbb N}$ and define the relational structure
${\bf D}=(D; R_{ij} \; | \; 1 \leq i,j \leq d)$ where $R_{ij}=\{(i,j)\}$. This structure clearly consists of binary JEPD relations but is technically not a partition scheme since it does not contain the equality relation over $D$.
We can, however, easily qffo-define equality in ${\bf D}$ via 
\[x =_D y \equiv R_{11}(x,y) \vee R_{22}(x,y) \vee \dots \vee R_{dd}(x,y).\]
Note then that
{\em any} constraint language $\Gamma$ with domain $D$ can be viewed as a qffo reduct over ${\bf D}$. 
Let $t=(t_1,\dots,t_k) \in D^k$. The relation $R_t$ that only contain the tuple $t$ can be defined as
\[R_t(x_1,\dots,x_k) \equiv \bigwedge_{i=1}^{k} R_{t_it_i}(x_i,x_i).\]
Any relation $R \subseteq D^k$ can now be defined as
\[R(x_1,\dots,x_k) \equiv \bigvee_{t \in R} R_t(x_1,\dots,x_k).\]
Hence, a backdoor in the style of Definition~\ref{def:fin_backdoor} is a special case of Definition~\ref{def:backdoor}, meaning that our backdoor notion is not merely an adaptation of the finite-domain concept, but a strict generalisation, since we allow arbitrary binary relations (and not only unary relations) in the underlying relational structure.
Hence, we only care whether each constraint, locally, can be simplified to the target language. 
\end{example}

We now have a working backdoor definition for infinite-domain CSPs, but it remains to show that they actually simplify CSP solving, and that they can be found efficiently. 
We study such computational aspects in the following section.

%We will now turn our attention to the complexity of {\em backdoor evaluation} and {\em backdoor detection}.

\subsection{Algorithms for Finite Languages}
\label{sec:algs}

This section is  divided into three parts where we analyse various computational problems associated with
backdoors. %Observe that we exclusively consider finite constraint languages in this section.
%We simplify the presentation by using the following convention: we write $[\source,\target,\base]$ to indicate that $\base$ is a base language and that $\source$ and $\target$
%are reducts of $\base$.

\subsubsection{Computing Simplification Maps}
\label{sec:simp-map}

We discussed (in Section~\ref{sec:basic-defs}) the fact that a simplification map from $\source$ to $\target$
always exists when $\source$ and $\target$ are finite languages. How to compute such a map
is an interesting question in its own right.
In the finite-domain case, the computation is straightforward (albeit time-consuming) since one can enumerate all
solutions to a CSP instance in finite time. This is clearly not possible when the domain
is infinite. Thus, we introduce a method that circumvents this difficulty by enumerating
other objects than concrete solutions.

%Let $\base$ be a (binary) base language (with domain $D$) and let $\source,\target$ be qffo-definable in $\base$.
%Say that the relations in $\base$ are {\em jointly exhaustive} (JE)
%if $\bigcup \base = D^2$, and they are
%{\em pairwise disjoint} (PD)
%if $R \cap R' = \varnothing$ for all distinct $R,R' \in \base$.

Assume that $[\source,\target,\base]$ is a backdoor triple.
%Let $I = (V, C)$ be an instance of CSP$(\target)$.
We first make the following observation concerning relations that are qffo-defined in a
partition scheme $\base$ (for additional details, see e.g. Sec. 2.2. in Lagerkvist \& Jonsson (2017)\nocite{lagerkvist2017d}). If $R$ is qffo-defined in $\base$ and $\base$ is finite, then
$R$ can be defined by a DNF $\base$-formula that involves only 
positive (i.e. negation-free) atomic formulas of type $R(\bar{x})$, where $R$ is a relation in $\base$: every
atomic formula $\neg R(x,y)$ can be replaced by
\[\bigvee_{S \in \base\setminus \{R\}} S(x,y)\]
and the resulting formula being transformed back to DNF.

Let $I=(V,C)$ denote an arbitrary instance of, for instance, CSP$(\source)$.
An $\base$-{\em certificate} for $I$ is a satisfiable instance 
${\mathcal C} = (V, C')$ of CSP$(\base)$ that {\em implies}
every constraint in $C$, i.e.
for every $R(v_1,\dots,v_k)$ in $C$,
there is a clause in the definition of this constraint
(as a DNF $\base$-formula) such that all literals
in this clause are in $C'$. 
It is not difficult to see that $I$ has a solution if and only if
$I$ admits an $\base$-certificate (see, for instance,
Theorem 6 by Jonsson and Lagerkvist (2017)\nocite{lagerkvist2017d} for
a similar result). Hence, we will sometimes also say that an $\base$-certificate ${\mathcal C}$ of a CSP$(\source)$ instance $I$ {\em satisfies} $I$.
We will additionally use {\em complete certificates}: a CSP$(\cdot)$ instance is {\em complete}
if it contains a constraint over every 2-tuple of (not necessarily distinct) variables, and
a certificate is complete if it is 
a complete instance of CSP$(\cdot)$.
%The reader should be aware of the fact that complete certificates are sometimes defined slightly
%differently for technical reasons~\cite{Dabrowski:etal:aaai2021jepd}.

\begin{example}
Consider the partition scheme $({\mathbb Q};<,>,=)$ based on 
the rationals under the natural ordering.
Let us once again consider the betweenness relation 
$\beta=\{(x,y,z) \in {\mathbb Q}^3 \; | \; x<y<z \vee z<y<x\}$
from Example~\ref{ex:eq}. Assume that
\[I=(\{x,y,z,w\} \; | \; \{\beta(x,y,z),\beta(y,z,w)\})\] 
is an instance of
CSP$(\{\beta\})$. The instance $I$ is satisfiable and this is witnessed
by the solution $f(x)=0,f(y)=1,f(z)=2,f(w)=3$.
A certificate for this instance is
$\{x<y,y<z,z<w\}$ and a complete certificate is
\[\begin{array}{lllllll}
x<y, & x < z, & x < w, &
y<z, & y<w, &
z<w, \\
y>x, & z>x, & w>x, &
z>y, & w>y, &
w>z, \\
x=x, & y=y, & z=z, & w=w. \\
\end{array}
\]
\end{example}

%In the next lemma, we will exploit a slightly
%stronger property that complete certificates based on JEPD languages have.
%An instance of CSP$(\cdot)$ is {\em complete}
%if it contains a constraint over every
%pair of distinct variables, and
%a certificate is complete if it is 
%a complete instance of CSP$(\cdot)$.

We first show that satisfiable CSP instances always have complete certificates under fairly
general conditions. Furthermore, every solution is covered by at least one such certificate.

\begin{lemma}  \label{lem:certificates}
Assume $\base$ is a partition scheme and that $\Gamma$ is qffo-definable in $\base$.
An instance $(V,C)$ of CSP$(\Gamma)$ has a solution $f \colon V \rightarrow D$ if and only if
there exists a complete $\base$-certificate for $I$ that has solution $f$.
\end{lemma}
\begin{proof}
 Let $I=(V,C)$ be an arbitrary instance of CSP$(\Gamma)$.

 \medskip

 Assume ${\mathcal C}=(V,\widehat{C})$ is a complete $\base$-certificate for $I$ with a solution $f \colon V \rightarrow D$.
 The certificate ${\mathcal C}=(V,\widehat{C})$ implies
 every constraint in $C$. Arbitrarily choose a constraint
 $R(v_1,\dots,v_k)$ in $C$.
 There is a clause in the definition of this constraint
 (viewed as a DNF $\base$-formula) such that all literals
 in this clause are in $\widehat{C}$. This implies that $(f(v_1),\dots,f(v_k)) \in R$
 since $f$ is a solution to ${\mathcal C}$. We conclude that $f$ is a solution to $I$
 since $R(v_1,\dots,v_k)$ was chosen arbitrarily.

 \medskip

 Assume $f:V \rightarrow D$ is a solution to $I$. 
 We construct a complete 
$\base$-certificate ${\mathcal C}=(V,\widehat{C})$ for $I$ such that $f$ is a solution
to ${\mathcal C}$.
 We know that $\base$ is JEPD.
 We construct a complete certificate ${\mathcal C}=(V,\widehat{C})$ such that $f$ is a solution to ${\mathcal C}$.
 Consider a 2-tuple of (not necessarily distinct) variables $(v,v')$ where $\{v,v'\} \subseteq V$.
 The tuple $(f(v),f(v'))$ appears in exactly one relation $R$ in $\base$ since $\base$ is JEPD.
 Add the constraint $R(v,v')$ to $\widehat{C}$. Do the same thing for all 2-tuples of variables.
 The resulting instance ${\mathcal C}$ is complete and it is
 satisfiable since $f$ is a valid solution.
 \end{proof}

We use the previous lemma for proving that two instances $I_s$ and $I_t$ have the same solutions if and only
if they admit the same complete certificates.

\begin{lemma} \label{lemma1-help}
Let $[\source,\target,\base]$ be a backdoor triple such that $\source$, $\target$, and $\base$ are finite.
Given instances $I_s=(V,C)$ of CSP$(\source)$ and $I_t=(V,C')$ of CSP$(\target)$, the following are
equivalent:

\begin{enumerate}
\item
$\sols(I_s) =  \sols(I_t)$ and

\item
$I_s$ and $I_t$ have the same set of complete $\base$-certificates.
\end{enumerate}
\end{lemma}
 \begin{proof}
 Arbitrarily choose an instance $I_s=(V,C_s)$ of CSP$(\source)$ and an instance $I_t=(V,C_t)$ of CSP$(\target)$.

 Assume that $I_s$ and $I_t$ have the same set of complete $\base$-certificates.
 Arbitrarily choose a solution $f:V \rightarrow D$ to $I_s$
 that is not a solution to $I_t$ (the other direction is analogous).
 There is a complete $\base$-certificate ${\mathcal C}$ for $I_s$ such that $f$ is a solution to ${\mathcal C}$ by 
 Lemma~\ref{lem:certificates}. We know that
 ${\mathcal C}$ is a certificate for $I_t$ so Lemma~\ref{lem:certificates} implies that $f$ is a solution to $I_t$, too.
 This leads to a contradiction.

 Assume that $I_s$ and $I_t$ have the same set of solutions.
 Assume ${\mathcal C}$ is a complete $\base$-certificate for $I_s$ but not for $I_t$ (the other
 way round is analogous).
 By Lemma~\ref{lem:certificates}, every solution to ${\mathcal C}$ is a solution to $I_s$.
 Since $I_s$ and $I_t$ have the same set of solutions, ${\mathcal C}$ is a complete $\base$-certificate for $I_t$, too,
 which leads to a contradiction.
 \end{proof}

Finally, we present our method to compute simplification maps.

\begin{lemma} \label{lemma1-c}
Let $X$ be the set of backdoor triples $[\source,\target,\base]$ that enjoy the following properties:

\begin{enumerate}
\item
$\source$, $\target$, and $\base$ are finite and

\item
CSP$(\base)$ is decidable.
\end{enumerate}

\noindent
The problem of constructing simplification maps for
members of $X$ is computable.
\end{lemma}
 \begin{proof}
 Arbitrarily choose $[\source,\target,\base]$ in $X$.
 Recall the definitions of ${\bf S}$ and ${\bf T}$ that were made
 in connection with Definition~\ref{def:simp-map}.
 Arbitrarily choose a relation $R \in {\bf S}$ with arity $k$ and 
 define $I_s=(V,C)=(\{v_1,\dots,v_k\},\{R(v_1,\dots,v_k)\})$. Given a $k$-ary formula $\varphi \in {\bf T}$, let
 $I_t=(V,\varphi(v_1,\dots,v_k))$. Then, the following are
 equivalent
 
 \begin{enumerate}
 \item[(a)]
 $\sols(I_s)=\sols(I_t)$, 

 \item[(b)]
 $I_s$ and $I_t$ have the same set of complete $\base$-certificates
 \end{enumerate}
 
 by Property 1 combined with Lemma~\ref{lemma1-help}.
 Property 2 implies that condition (a) is decidable: there is a straightforward algorithm based on
 enumerating all complete $\base$-certificates.
 Compute every possible complete $\base$-certificate on the variables in $V$ and check whether
 $I_s$ and $I_t$ are equisatisfiable on these certificates. Recall that checking
 if a certificate implies the constraints in $I_s$ and $I_t$ is a decidable problem since
 CSP$(\base)$ is decidable.
 This procedure can be performed in a finite number of steps since
 the number of complete $\base$-certificates on variable set $V$ is finite.

 With this in mind, there is an algorithm that computes a simplification map $\Sigma \colon \source \rightarrow \target$.
 Arbitrarily choose a $k$-ary relation $R$ in ${\bf S}$ and
 let $I_s=(V,C_s)=(\{v_1,\dots,v_k\},\{R(v_1,\dots,v_k)\})$.
 Enumerate all $I_t=(V,\varphi(v_1,\dots,v_k))$ where $\varphi \in {\bf T}$ is $k$-ary.
 If there exists an $I_t$ that satisfies condition (a), then let $\Sigma(R)=\varphi$, and, otherwise, let $\Sigma(R)$ be undefined.
 We know that testing condition (a) is decidable by Property 2 and we know that ${\bf S}$ and ${\bf T}$ are finite sets,
 so $\Sigma$ can be computed in a finite number of steps.
 \end{proof}

%Decidable CSPs based on JEPD relations are abundant in the literature and include
%fundamental formalisms like Allen's algebra, the RCC formalisms, and various cardinal direction calculi:
%we refer the reader to the survey by Dylla et al.~\cite{Dylla:2017:SQS:3058791.3038927}.

\subsubsection{Backdoor Evaluation}
\label{sec:backdoor-eval}

We begin by studying the complexity of the following problem, which intuitively, 
says to which degree the existence of a
  backdoor helps to solve the original problem. 
\ifshort
\begin{definition} ($[\source,\target,\base]$-\textsc{Backdoor Evaluation}) Given $(V,C)$
  of CSP$(\source)$ and a backdoor $B \subseteq V^2$ into CSP$(\target)$, determine if $(V,C)$ is satisfiable.
\end{definition}
\fi
\iflong
\pbDef{$[\source,\target,\base]$-\textsc{Backdoor Evaluation}}{A CSP instance $(V,C)$
  of CSP$(\source)$ and a backdoor $B \subseteq V^2$ into CSP$(\target)$.}
  {Is $(V,C)$ satisfiable?}
\fi

  %\textsc{($\Gamma, \Gamma')$-backdoor evaluation}\\
  %{\bf Input:} A CSP instance $(V,C)$
  %of CSP$(\Gamma)$ and a backdoor $B \subseteq V^2$ into CSP$(\Gamma')$.\\
  %{\bf Question:} 
  %Let $A$ contain all functions $\alpha:B \rightarrow \base$.
  %Is there an $\alpha \in A$ such that $\Sigma((V,C|_{\alpha}))$ 
  %is a yes-instance of CSP$(\Gamma')$.
  
%Clearly, \textsc{$[\source,\target,\base]$-backdoor evaluation} is in general NP-hard: simply pick an equality language $\source$ such %that CSP$(\source)$ is NP-hard. One may, for instance choose $\source = \{S\}$ where $S$ is the relation
%presented in Example~\ref{ex:eq}. 
%Then \textsc{$[\source$, $\{=,\neq\},  \{=, \neq\}]$)-backdoor evaluation} is NP-hard by a simple reduction from CSP$(\source)$ with the %set $B = V^2$, which is trivially a backdoor with respect to the target language $\{=,\neq\}$ since the relations between all pairs of %variables are fully determined in each reduced instance. 
  
  Clearly, \textsc{$[\source,\target,\base]$-backdoor evaluation} is in many cases \NP-hard: simply pick a language $\source$ such that CSP$(\source)$ is \NP-hard.
  Note that one, strictly speaking, is not forced to use the backdoor when solving the \textsc{$[\source,\target,\base]$-Backdoor Evaluation} problem, but if the size of the backdoor is sufficiently small then we may be able to solve the instance faster via the backdoor.
  Indeed, as we will now prove, the problem is in $\FPT$ for finite languages when parameterised by the size of the backdoor.

\begin{theorem} \label{thm:backdoor-eval}
Assume that $[\source,\target,\base]$ is a backdoor triple such that $\source$, $\target$ and $\base$ are finite and CSP$(\target)$ and CSP$(\base)$ are polynomial-time solvable.
  Then, \textsc{$[\source,\target,\base]$-Backdoor Evaluation} is solvable in $\bigoh(|\base|^k \cdot {\rm poly}(||I||))$ where $k$ is the size of the backdoor.  Hence, the problem is in $\FPT$ when parameterised by $k$.

\end{theorem}  

\begin{proof}
  Let $\Sigma \colon \source \rightarrow \target$ be a simplification map that has been computed off-line, let $I = (V,C)$ be an instance of CSP$(\source)$, let
  $B \subseteq V^2$ be a backdoor of size $k$, and let $m = |\base|$. 
  %We assume without loss of generality that the
  %simplification map has been computed off-line. 
  Then, we claim that $I$ is satisfiable if and only if there is a consistent assignment $\alpha \colon B \rightarrow \base$
  such that the CSP$(\target)$ instance $I_{\mid \alpha} = (V, \{\Sigma(c_{\mid \alpha}) \mid c \in C\})$ is satisfiable.
  
  \smallskip
  
  \noindent
  {\em Forward direction.}
  Assume that $I$ is satisfiable. Since $\base$ is JEPD, and since $\source$ is qffo-definable in $\base$, we know from Lemma~\ref{lem:certificates} that $I$ admits a complete certificate $(V, \widehat{C})$. For every pair $(x,y) \in B$ then define $\alpha$ to agree with the complete certificate $(V, \widehat{C})$, i.e., $\alpha(x,y) = S$ for $S(x,y) \in \widehat{C}$. Naturally, $\alpha$ is consistent since $(V, \widehat{C})$ is a complete certificate for $I$, and since $B$ is a backdoor set it also follows that the CSP$(\target)$ instance $(V, \{\Sigma(c_{\mid \alpha}) \mid c \in C)$ is well-defined. Pick an arbitrary constraint $\Sigma(R(x_1, \ldots, x_{\ar(R)})_{\mid \alpha})$. It follows (1) that $(V, \widehat{C})$ satisfies $R(x_1, \ldots, x_{\ar(R)})$, and (2) that if $\alpha(x_i, x_j) = S$ for $x_i, x_j \in \{x_1, \ldots, x_{\ar(R)}\}$ then $(V, \widehat{C})$ satisfies $S(x_i, x_j)$, meaning that $(V, \widehat{C})$ satisfies 
  \[R(x_1, \ldots, x_{\ar(R)}) \land \bigwedge_{\alpha(x_i, x_j) = S, x_i, x_j \in \{x_1, \ldots, x_{\ar(R)}\}}S(x_i, x_j),\]
  and hence also $\Sigma(R(x_1, \ldots, x_{\ar(R)})_{\mid \alpha})$, since $\Sigma$ is a simplification map. 
  
  \smallskip
  
  \noindent
  {\em Backward direction.}
  Assume that there exists a consistent $\alpha \colon B \rightarrow \base$ such that $(V, \{\Sigma(c_{\mid \alpha}) \mid c \in C\})$ is satisfiable, and let $(V, \widehat{C})$ be a complete certificate witnessing this. Naturally, for any pair $(x,y) \in B$ it must then hold that $S(x,y) \in \widehat{C}$ for $\alpha(x,y) = S$, since $(V, \widehat{C})$ could not be a complete certificate otherwise. Pick a constraint $R(x_1, \ldots, x_{\ar(R)}) \in C$, and let $\Sigma(R(x_1, \ldots, x_{\ar(R)})_{\mid \alpha}) \equiv \varphi(x_1, \ldots x_{\ar(R)})$ for some $\varphi(x_1, \ldots, x_{\ar(R)}) \in {\bf T}$. It follows that $(V, \widehat{C})$ satisfies 
  \[\varphi(x_1, \ldots, x_{\ar(R)})\]
  and since 
  \[\sols(\varphi(x_1, \ldots, x_{\ar(R)})) = \sols(R(x_1, \ldots, x_{\ar(R)})_{\mid \alpha})\]
  it furthermore follows that $R(x_1, \ldots, x_{\ar(R)})_{\mid \alpha}$ must be satisfied, too. However, since \[R(x_1, \ldots, x_{\ar(R)})_{\mid \alpha} \equiv R(x_1, \ldots, x_{\ar(R)}) \land \bigwedge_{\alpha(x_i, x_j) = S, x_i, x_j \in \{x_1, \ldots, x_{\ar(R)}\}}S(x_i, x_j),\] and since every constraint $S(x_i, x_j)$ is clearly satisfied, it must also be the case that $(V, \widehat{C})$ is a complete certificate of $I$.

\smallskip

  Put together, it thus suffices to enumerate all
  $m^k$ choices for $\alpha$ and to check whether $\alpha$ is consistent and whether $I_{\mid \alpha}$ is satisfiable. CSP$(\base)$ is tractable so checking whether $\alpha$ is consistent can be done in polynomial time. Moreover, using the simplification map $\Sigma$ (which is constant-time accessible since $\source$ and $\target$ are finite), we can reduce
  $I_{\mid \alpha}$ to an instance of CSP$(\target)$, which can be solved in polynomial-time.
  The total running
  time is $O(m^k \cdot \mathrm{poly}(||I|||)$.
\end{proof}

We remark that the assumption that CSP$(\base)$ is in P is only used to verify whether $\alpha$ is consistent, and this can sometimes be achieved even when CSP$(\base)$ is NP-hard; one such case is when $\base$ is finitely bounded and homogeneous (see Example~\ref{ex:finbound}).

%\iflong
%Clearly, \textsc{$[\source,\target,\base]$-backdoor evaluation} is in general NP-hard: simply pick an equality language $\source$ such %that CSP$(\source)$ is NP-hard. Then \textsc{$[\source$, $\{=,\neq\},  \{=, \neq\}]$)-backdoor evaluation} is NP-hard by a simple %reduction from CSP$(\source)$ with the set $B = V^2$, which is trivially a backdoor with respect to the target language $\{=,\neq\}$ %since the relations between all pairs of variables are fully determined in each reduced instance. 
%\fi
%% Given a backdoor of size $k \leq |V|^2$ we can then solve $(V,C)$ in
%% $O(m^k \cdot \mathrm{poly}(||I|))$ time by enumerating all possible
%% functions from $B$ to the set of basic relations $\{R_1, \ldots,
%% R_m\}$. 

\subsubsection{Backdoor Detection}
\label{sec:backdoordetection}

Theorem~\ref{thm:backdoor-eval} implies that
small backdoors are desirable since they can be used to solve CSP problems faster. Therefore, let us now turn to the problem of finding backdoors. The basic backdoor detection problem is defined as follows. 
\ifshort
\begin{definition}
(\textsc{$[\source,\target,\base]$-Backdoor Detection}) Given an instance $(V,C)$ of CSP$(\source)$ and an integer $k$, determine whether $(V,C)$ has a backdoor $B$ into $\target$ of size at most $k$. %(and if so output such a backdoor).
\end{definition}
\fi
\iflong
\pbDef{\textsc{$[\source,\target,\base]$-Backdoor Detection}}{A CSP instance $(V,C)$ of CSP$(\source)$ and an integer $k$.}{Does $(V,C)$ have a backdoor $B$ into $\target$ of size at most $k$? (and if so output such a backdoor).}
\fi

The problem is easily seen to be \NP-hard even when $\source$ and $\target$ are finite; we will provide a proof of this 
in Corollary~\ref{cor:finNP}.
We will now prove that the problem can be solved efficiently if the size of the backdoor is sufficiently small.

\begin{theorem}
  
  %\textsc{$[\source,\target,\base]$-Backdoor Detection} is in FPT when parameterized by $k$, if $\source$, $\target$ and $\base$ %are finite, and CSP$(\target)$ and CSP$(\base)$ are tractable.
  
  Assume that $[\source,\target,\base]$ is a backdoor triple such that $\source$, $\target$, $\base$ are finite and CSP$(\target)$ and CSP$(\base)$ are polynomial-time solvable.
  Then, \textsc{$[\source,\target,\base]$-Backdoor Detection} is solvable in $\bigoh(\binom{a}{2}^{k+1} \cdot |\base|^k \cdot {\rm poly}(||I||)) $ time where $a$ is the maximum arity of the relations in $\source$ and $k$ is the size of the backdoor.  Hence, the problem is in $\FPT$ when parameterised by $k$.
\end{theorem}
\begin{proof}
    Let $I=((V,C),k)$ be an instance of \textsc{$[\source,\target,\base]$-Backdoor Detection}, let $a$ be the maximum arity of any relation in $\source$, and let $\Sigma$ be a simplification map from $\source$ to $\target$ which we assume has been computed off-line. We solve $((V,C),k)$
    using a bounded depth search tree algorithm as follows. 
    
    We construct a search tree $T$, for which
    every node is labeled by a set $B \subseteq V^2$ of size at most $k$. 
    Additionally, every leaf node has a second label,
    which is either \textsc{Yes} or \textsc{No}. 
    $T$ is defined inductively as follows. The root of $T$ is labeled by
    the empty set. Furthermore, if $t$ is a node of $T$,
    whose first label is $B$, then the
    children of $t$ in $T$ are obtained as follows. If for every consistent 
    assignment $\alpha \colon B \rightarrow \base$, where $\base=\{R_1, \ldots, R_m\}$, and every $c \in C$, we have that $\Sigma(c_{\mid \alpha})$ is defined, then $B$ is a backdoor into $\target$
    of size at most $k$ and therefore $t$ becomes a leaf node, whose second label is \textsc{Yes}.
    Otherwise, i.e., if there is a consistent assignment $\alpha \colon B \rightarrow \base$ and a constraint $c \in C$ such that 
    $\Sigma(c_{\mid \alpha})$ is not defined, we distinguish two cases:
    (1) $|B|=k$, then $t$ becomes a leaf node, whose second label is \textsc{No}, and (2) $|B|<k$, then
    for every pair $p$ of variables in the scope of $c$ with $p \notin B$, $t$ has a child whose
    first label is $B\cup \{p\}$.
    
    If $T$ has a leaf node, whose second label is \textsc{Yes}, then the
    algorithm returns the first label of that leaf node. Otherwise the
    algorithm returns \textsc{No}. This completes the description of the
    algorithm.

    We now show the correctness of the algorithm. First, suppose 
    the search tree $T$ built by the algorithm has a leaf node $t$ whose
    second label is \textsc{Yes}. Here, the algorithm returns the first label, say
    $B$ of $t$. By definition, we obtain that $B$ is a backdoor into $\target$ of size at most $k$.
    
    Now consider the case where the algorithm returns
    \textsc{No}. We need to show that there is no backdoor set $B$ into $\target$ with $|B|\leq k$.
    Assume, for the sake of contradiction that such a set $B$ exists.

    Observe that if $T$ has a leaf node
    $t$ whose first label is a set $B'$ with $B' \subseteq B$,
    then the second label of $t$ must be \textsc{Yes}. This is because,
    either $|B'|<k$ in which case the second label of $t$ must be
    \textsc{Yes}, or $|B'|=k$ in which case $B'=B$ and by the definition
    of $B$ it follows that the second label of $t$ must be \textsc{Yes}.

    It hence remains to show that $T$ has a leaf node whose first label is a set $B'$
    with $B' \subseteq B$. This will complete the proof about the
    correctness of the algorithm. We will show a slightly stronger
    statement, namely, that for every natural number $\ell$, either $T$
    has a leaf whose first label is contained in
    $B$ or $T$ has an inner node of distance exactly $\ell$ from the root
    whose first label is contained in $B$. We show the latter by
    induction on $\ell$.

    The claim obviously holds for $\ell=0$. So assume that $T$ contains a
    node $t$ at distance $\ell$ from the root of $T$ whose first label, say
    $B'$, is a subset of $B$. 
    If $t$ is a leaf node of $T$, then the
    claim is shown. Otherwise, there is a consistent assignment $\alpha' \colon B' \rightarrow \base$ 
    and a constraint $c \in C$ such that 
    $\Sigma(c_{\mid \alpha'})$ is not defined.
    
    Let $\alpha \colon B \rightarrow \base$ be any
    consistent assignment of the pairs in $B$ that agrees with $\alpha'$ on the pairs in $B'$.
    Then, $\Sigma(c_{\mid \alpha})$ is defined because $B$ is a backdoor set into $\target$.
    By definition of the search
    tree $T$, $t$ has a child $t'$ for every pair $p$ of variables in the scope
    of some constraint $c \in C$ such that $\Sigma(c_{\mid \alpha'})$ is not defined. 
    We claim that $B$ contains at least one pair of variables within the scope of c. Indeed,
    suppose not. Then $\Sigma(c_{\mid \alpha})=\Sigma(c_{\mid \alpha'})$ and this contradicts our assumption
    that $\Sigma(c_{\mid \alpha})$ is defined.
    This concludes our proof concerning the correctness of the algorithm.

  The running time of the algorithm is obtained as
  follows. Let $T$ be a search tree obtained by the
  algorithm. Then the running time of the depth-bounded search tree
  algorithm is $O(|V(T)|)$ times the maximum time that is spent on any
  node of $T$. Since the number of children of any node of $T$ is
  bounded by $\binom{a}{2}$ (recall that $a$ is the maximum 
  arity of the relations in of $\source$) and the
  longest path from the root of $T$ to some
  leaf of $T$ is bounded by $k+1$, we obtain that $|V(T)| \leq
  \bigoh({(\binom{a}{2})}^{k+1})$. Furthermore, the time
  required for any node $t$ of $T$
  is at most
  $\bigoh(m^k|C| \cdot \mathrm{poly}(||I||)$ (where the polynomial factors stems from checking whether $\alpha$ is consistent---keep in mind that $\Sigma$ is constant-time accessible since $\source$ and $\target$ are finite). Therefore we obtain $\bigoh((\binom{a}{2})^{k+1}m^k|C|)$ as the total run-time of the algorithm showing
  that \textsc{$[\source,\target,\base]$-Backdoor Detection} is fpt when parameterized by $k$.
  %\hl{I added the consistency check (and added tractability of the target language as an assumption in the theorem.}
\end{proof}

Again, we remark that the assumption that CSP$(\base)$ is in P is only used to verify whether $\alpha$ is consistent (see the comment after Theorem~\ref{thm:backdoor-eval}).

\subsection{Hardness Results for Infinite Languages}
\label{sec:inf}

We will now see that the positive fpt results from the two previous sections are, unfortunately, restricted to finite languages: finiteness is not merely a simplifying assumption, but in many cases absolutely crucial for tractability. 
We remind the reader that if the source language is infinite, then there are an infinite number of possible inputs for
the simplificaton map, and this implies that it is not necessarily accessible in polynomial time.
However, we will see that simplification maps with good computational properties
do exist in certain cases. Even under this assumption, we prove that
the backdoor detection problem is in general \Weft[2]-hard. We do not
study the backdoor evaluation problem since the hardness of backdoor detection makes
the evaluation problem less interesting. We note that analysing CSPs with infinite constraint languages is often problematic since the way
in which the input is encoded may influence
the complexity. We will therefore consider concrete, infinite languages in the sequel, and will for all involved relations explicitly state their defining formulas, making it easy to represent the associated computational problems.

We begin by establishing the existence of a relation which turns out to be useful as a gadget in the forthcoming hardness reduction.
For every $k \geq 2$, we let the $k$-ary equality relation $R_k$ be defined as follows:
\[R_k(x_1,\dots,x_k) \equiv \bigwedge_{\mbox{\scriptsize{$i,j,l,m \in [k]$ with $i\neq j$ and $l \neq m$}}}   (x_i\neq x_j \lor x_l=x_m)   \]
We illustrate this construction with $R_3(x_1,x_2,x_3)$:
\[(x_1 \neq x_2 \vee x_1 = x_2) \wedge (x_1 \neq x_2 \vee x_1=x_3) \wedge (x_1 \neq x_2 \vee x_2=x_3) \; \wedge\]
\[(x_1 \neq x_3 \vee x_1 = x_2) \wedge (x_1 \neq x_3 \vee x_1=x_3) \wedge (x_1 \neq x_3 \vee x_2=x_3) \; \wedge\]
\[(x_2 \neq x_3 \vee x_1 = x_2) \wedge (x_2 \neq x_3 \vee x_1=x_3) \wedge (x_2 \neq x_3 \vee x_2=x_3)\]
or, slightly simplifed,
\[ (x_1 \neq x_2 \vee x_1=x_3) \wedge (x_1 \neq x_2 \vee x_2=x_3) \; \wedge\]
\[(x_1 \neq x_3 \vee x_1 = x_2) \wedge (x_1 \neq x_3 \vee x_2=x_3) \; \wedge\]
\[(x_2 \neq x_3 \vee x_1 = x_2) \wedge (x_2 \neq x_3 \vee x_1=x_3).\]
Now, note that $R_3(x_1,x_2,x_3) \wedge (x_1=x_2)$ is equivalent to
\[(x_1=x_3) \wedge (x_2=x_3)\]
and $R_3(x_1,x_2,x_3) \wedge (x_1 \neq x_2)$ is equivalent to
\[(x_1 \neq x_2) \wedge (x_1 \neq x_3) \wedge (x_2 \neq x_3).\]
The important properties of $R_k$ are summarised in the next lemma.

\begin{lemma}  \label{lem:eqrel}
    Let $\base$ be a partition scheme with infinite domain $D$. Then the following holds for every $k \geq 2$.
    \begin{enumerate}
        \item $R_k(x_1,\ldots,x_k)$ cannot be qfpp-defined over  $\base \cup \{\neq\}$, and
        \item for every pair $i$, $j$ with $1 \leq i < j \leq k$ and every assignment $\alpha$ of $(x_i,x_j)$ to 
        $\base$, it holds that $R'_k(x_1, \ldots, x_k) \equiv R_k \land \alpha((x_i,x_j))(x_i, x_j)$ can be qfpp-defined over $\base \cup \{\neq\}$.
    \end{enumerate}
    Moreover, the definition of $R_k$ can be computed in time $k^4$.
\end{lemma}
\begin{proof}
     %Consider the relation $R_k$ given by the conjunction of $(x_i\neq x_j \lor x_l=x_m)$ for every
     %$i,j,l,m \in [k]$ with $i\neq j$ and $l \neq m$.
     For proving the first statement,
     we begin by showing that $(a_1,\dots,a_k) \in R_k$ if either (1)
     $a_1=a_2=\ldots=a_k$ or (2) $a_i\neq a_j$ for 
             every $i$ and $j$ with $i\neq j$. Assume this is not the case. Then there are $i,j,m,l\in[k]$ with $i\neq j$ and $m\neq l$ such that $a_i=a_j$ and $a_l\neq a_m$. But then the term $(x_i\neq x_j \lor x_l=x_m)$ in the definition of $R_k$ is not satisfied by $(a_1,\dots,a_k)$. Now, consider a conjunction $\phi$ of atomic formulas from the set $\{R(x_i, x_j) \mid R \in \base \cup \{\neq\}, i,j \in [k]\}$. If each atomic formula in $\phi$ is 
             of the type $x_i=x_j$, then the models of $\phi$ cannot correctly define $R_k$: $\phi$ is not satisfied by any assignment where all variables are assigned distinct values. Similarly, if there exists an atomic formula of the type $R(x_i, x_j)$
             in $\phi$ where $R$ is not the equality relation, then $\phi$ cannot be satisfied by an assignment where all variables are assigned the same value, since $R$ by assumption is pairwise disjoint with the equality relation. Hence, $R_k$ cannot be qfpp-defined by $\base$.
    
     For the second statement, let $\alpha$ be an assignment of a pair $(x_i,x_j)$ to some $R \in \base$.
     We observe the following.
     \begin{itemize}
         \item If $R$ is the equality relation, then $R_k(x_1,\dots,x_k) \land R(x_i,x_j)$ is logically equivalent
           to the formula $(x_1=x_2)\land (x_1=x_3) \land \ldots \land (x_1=x_k)$. The definition
           of $R_k$ contains the clauses $(x_i \neq x_j \vee x_l = x_m)$ for all $1 \leq l \neq m \leq k$.
           Since $x_i \neq x_j$ does not hold due to $R(x_i,x_j)$, it follows that all variables must be assigned the same value.
          
         \item If $R$ is not the equality relation, then $R_k(x_1,\dots,x_k) \land R(x_i,x_j)$ is logically equivalent to the conjunction of $(x_i\neq x_j)$ for every $i,j \in [k]$ where $i\neq j$. To see this, note that the definition
           of $R_k$ contains the clauses $(x_i \neq x_j \vee x_l = x_m)$ for all $1 \leq l \neq m \leq k$.
           Since $x_i = x_j$ does not hold due to $R(x_i,x_j)$ (again, recall that $R$ by assumption is pairwise disjoint with the equality relation), it follows that all variables must be assigned distinct values.
     \end{itemize}
     We finally note that the definition of $R_k$ can easily be computed in $k^4$ time so
     $R_k$ satisfies the statement of the lemma.
 \end{proof}

Assume that $\base_e$ is a partition scheme. Let $\source_e = \{R_i \mid i \geq 1\}$ where $R_i$ is defined as in Lemma~\ref{lem:eqrel}, and let $\target_e= \base_e \cup \{\neq\}$.
Note that both $\source_e$ and $\target_e$ are qffo reducts of $\base_e$ regardless of the precise choice of $\base_e$.
We first verify that $\source_e$, despite being infinite, admits a straightforward simplification map to  the target language $\target_e= \base_e \cup \{\neq\}$.
%obtained by  letting $\Sigma(R_k(x_1, \ldots, x_k)_{\mid \alpha})$ be undefined if (1) $\alpha(x_i, x_j)$ is not defined 
%for any distinct $x_i, x_j \in \{x_1, \ldots, x_k\}$ and if all variables $x_1, \ldots, x_k$ are distinct, 
%and (2) otherwise $\Sigma(R_k(x_1, \ldots, x_k)_{\mid \alpha})$ is mapped to a suitable CSP$(\target_e)$ instance as prescribed by Lemma~\ref{lem:eqrel}. 

\begin{lemma} \label{lem:simpmapconstruction}
There is a simplification map $\Sigma_e$ from $\source_e$ to $\target_e$ that can be accessed in polynomial time.
\end{lemma}
 \begin{proof}
 Consider $\Sigma_e(R_k(x_1, \ldots, x_k)_{\mid \alpha})$. If $|\{x_1,\dots,x_k\}| < k$, then we may (without loss of generality)
 assume that we want to compute $\Sigma_e(R_k(x_1, x_1, x_2 \ldots, x_{k-1})_{\mid \alpha})$.
 This is equivalent to compute $\Sigma_e(R_k(x_1, y, x_2, \ldots, x_{k-})_{\mid \alpha})$ where $y$
 is a fresh variable and $\alpha$ is extended to $\alpha'$ so that $\alpha'(x_1,y)$ implies $x_1 = y$.
 Then, we map $\Sigma_e(R_k(x_1, y, \ldots, x_k)_{\mid \alpha'})$ to a suitable CSP$(\target_e)$ instance as prescribed by Lemma~\ref{lem:eqrel}.
 
 For the other case, assume instead that $|\{x_1,\dots,x_k\}| = k$.
 We let $\Sigma_e(R_k(x_1, \ldots, x_k)_{\mid \alpha})$ be undefined if 
 $\alpha(x_i, x_j)$ is not defined for any distinct $x_i, x_j \in \{x_1, \ldots, x_k\}$ --- this is justified by Lemma~\ref{lem:eqrel}. 
 Otherwise, we map $\Sigma_e(R_k(x_1, \ldots, x_k)_{\mid \alpha})$ to a suitable CSP$(\target_e)$ instance as prescribed by Lemma~\ref{lem:eqrel}. We conclude the proof by noting that these computations are easy to perform in polynomial time so $\Sigma_e$ is trivially polynomial-time accessible.
 \end{proof}

Our reduction is based on the following problem.

\pbDef{\textsc{Hitting Set}}{A finite set $U$, a family $\mathcal{F}$ of subsets of $U$, and an integer $k \geq 0$.}{Is there a set $S \subseteq U$ of size at most $k$ such that $S \cap F \neq \emptyset$ for every $F \in \mathcal{F}$?}

\textsc{Hitting Set} is \NP-hard even if the sets in $\mathcal{F}$ are restricted to sets of size 2: in this case, the problem is
simply the Vertex Cover problem. Furthermore, Hitting Set is W$[2]$-hard when parameterized by $k$~\cite{rod1999parameterized}
but this does not hold if the sets in $\mathcal{F}$ have size bounded by some constant.

\begin{theorem}\label{the:bd-wh}
\textsc{$[\source_e, \target_e, \base_e$]-backdoor detection} is W$[2]$-hard when parameterised by the size of the backdoor.
\end{theorem}

%Let $(V,C)$ be an instance of \eqCSP and $k$ an integer. Then, the problem of deciding whether $(V,C)$
%has a backdoor set of size $k$ into the tractable class \beqCSP{} is W$[2]$-hard parameterized by $k$.
%\begin{theorem} \label{thm:backdoorinfdomhard}
  %Let $(V,C)$ be an instance of \eqCSP and $k$ an integer. Then, the problem of deciding whether %$(V,C)$
  %has a backdoor set of size $k$ into the tractable class \beqCSP{} is W$[2]$-hard parameterized by %$k$.
%\end{theorem}
\begin{proof}
  We give a parameterized reduction from the Hitting Set problem.
  %That is, given a universe $U$ and a family $\mathcal{F}$ of subsets of $U$, the hitting set problem is the problem to decide whether %there is a (hitting) set $S\subseteq U$ of size at most $k$ such that $S\cap F\neq \emptyset$ for every $F \in \mathcal{F}$.
  Given an instance $(U,\mathcal{F},k)$ of Hitting set, let $(V,C)$ be the \eqCSP{} instance with
  $V=U\cup \{n\}$ having one constraint $C_F$ for every $F \in \mathcal{F}$, whose scope is $F\cup \{n\}$ and whose relation is $R_{|F|+1}$ (as defined in connection with Lemma~\ref{lem:eqrel}). This can easily be accomplished in polynomial time.
  %since $R_{|F| + 1}$ can be constructed in $O((|F| + 1)^4)$ time. 
  Next, we verify that $(U,\mathcal{F},k)$ has a hitting set of size at most $k$ if and only if $(V,C)$ has a backdoor set of size at most $k$ into \beqCSP{}.
  
  {\em Forward direction.} Let $S$ be a hitting set for $\mathcal{F}$. 
  We claim that $B=\{ (n,s) | s \in S\}$ is a backdoor set into \beqCSP{}.
  Because $S$ is a hitting set for $\mathcal{F}$, $B$ contains at least two variables from the scope of every constraint in $C$.
  Let $\alpha: B \rightarrow \base_e$ be an arbitrary consistent assignment. Arbitrarily choose a constraint
  $R_k(x_1,\dots,x_k)$ in $C$.
  By the construction of the simplification map (Lemma~\ref{lem:simpmapconstruction}), it follows that
  $\Sigma_e(R_k(x_1,\dots,x_k)_{\mid \alpha})$ is defined, so $B$
  is indeed a backdoor.

  {\em Backward direction.} Let $B$ be a backdoor set for $(V,C)$ into \beqCSP.
  Note first that we can assume that $b = (x,n)$ or $b = (n,x)$ for every $b \in B$. To see this, note that if this
  is not the case for some $b \in B$, then we can replace one of the variables in $b$ with $n$, while still obtaining a backdoor set, since it is sufficient to fix a single relation between pairs of variables in $R_k$ in order to simplify to \beqCSP{}.
  We claim that $(\bigcup_{b\in B}b)\setminus \{n\}$ is a hitting set for $\mathcal{F}$. This is clearly the case because for every constraint in $C$, there must be at least one pair $b \in B$ such that both variables in $b$ are in the scope of the constraint. Otherwise, there would exist a constraint whose simplification is the constraint itself, and such a constraint cannot be expressed as a conjunction of $\base_e$ constraints, due to the first condition of Lemma~\ref{lem:eqrel}.
\end{proof}
 
One may note that CSP$(\source_e)$ is polynomial-time solvable 
since $(0,\dots,0) \in R_k$ for all $k$.
Thus,
the \textsc{$[\source_e,\target_e,\base_e]$-backdoor detection} problem is
computationally harder than
the CSP problem that we attempt to solve with the backdoor approach. This indicates that the backdoor approach
must be used with care and it is, in particular, important to know the computational complexity of the CSPs under
consideration. Certainly, there are also examples of infinite source languages with an \NP-hard CSP such that
backdoor detection is \Weft[2]-hard.
For instance, let $\source'_e=\source \cup \{S\}$ where $S$ is the relation defined in Example~\ref{ex:eq}---it 
follows immediately that CSP$(\source'_e)$ is \NP-hard. Furthermore, it is not hard to verify that
Lemma~\ref{lem:simpmapconstruction} can be extended to the source language $\source'_e$ so the proof
of Theorem~\ref{the:bd-wh} implies \Weft[2]-hardness of \textsc{$[\source'_e,\target_e,\base_e]$-backdoor detection}, too.

\begin{corollary}
Let $\base$ denote a partition scheme with infinite domain $D$.
There exists an infinite constraint language $\source$ and a finite language $\target$
that are qffo definable
in $\base$ such that \textsc{$[\source,\target,\base]$-backdoor detection} is
\Weft[2]-hard.
\end{corollary}

Finally, we can now answer the question (that was raised in Section~\ref{sec:backdoordetection}) concerning the complexity of  \textsc{$[\source, \target, \base$]-backdoor detection}
when $\source$ and $\target$ are finite.
By observing that
the reduction employed in Theorem~\ref{the:bd-wh} is a polynomial-time reduction from Hitting Set and using the fact
that Hitting Set is \NP-hard even if all sets have size at most 2, we obtain the following result.
\begin{corollary}\label{cor:finNP}
    The problem
    \textsc{$[\{R_3\},\target_e,\base_e]$-Backdoor Detection} is \NP-hard.
\end{corollary}

\section{Sidedoors}
\label{sec:sidedoors}

We introduce the notion of {\em sidedoors} in this section. The basic idea is 
illustrated via an example in Section~\ref{sec:motex}. The example points out an intrinsic
problem with the backdoor approach,
and that the sidedoor approach has the
potential of being faster than the backdoor approach in certain cases.
We continue by formally defining sidedoors in Section~\ref{sec:sidedoor-def}.
In analogy with simplification maps for backdoors, our sidedoor definition is based on the idea of {\em branching maps} and
we discuss the construction of such maps in Section~\ref{sec:branch-map}.
We finally prove that the sidedoor evaluation and detection problems are 
fixed-parameter tractable when parameterised by solution size (Sections~\ref{sec:sidedoor-eval} and~\ref{sec:sidedoor-detection}, respectively).

\subsection{Motivating Example}
\label{sec:motex}

We use RCC-5 as the basis for this example so $\Theta=\{{\sf DR},{\sf
  PO},{\sf PP},{\sf PP}^{-1},{\sf EQ}\}$ denotes the set of basic
relations in RCC-5 and $\Theta^{\vee =}$ is the set of all
relations definable by unions of the basic relations. The problem CSP$(\Theta)$
is tractable while CSP$(\Theta^{\vee =})$ is \NP-complete~\cite{DBLP:journals/ai/RenzN99}.
Now, recall Example~\ref{ex:rcc5}.
Consider a reduced constraint $R_{\mid \alpha}(x,y)$ with respect to
an instance $(V,C)$ of CSP$(\Theta^{\vee =})$, a set $B \subseteq
V^2$, and a function $\alpha \colon B \rightarrow \Theta$. We know that
if $(x,y) \in B$ (or, symmetrically, $(y,x) \in B$) then
$R(x,y) \land (\alpha(x,y))(x,y)$ is either

\begin{enumerate}
\item
unsatisfiable if
$\alpha(x,y) \cap R = \emptyset$, or 

\item
equivalent to
$\alpha(x,y)$. 
\end{enumerate}

This implies that
the simplification map either
outputs an unsatisfiable CSP$(\Theta)$ instance or replaces the
constraints with the equivalent constraint over the basic relations in $\Theta$. 
Furthermore, this implies that a backdoor $B \subseteq V^2$ has to cover
{\em every} constraint in the instance which is not already included
in $\Theta$. This
results in an $O(5^{|B|})\cdot \mathrm{poly}(||I||)$ time algorithm
for RCC-5, which can be slightly improved to $O(4^{|B|})\cdot
\mathrm{poly}(||I||)$ with the observation that only the trivial
relation $({\sf DR},{\sf PO},{\sf PP},{\sf PP}^{-1},{\sf EQ})$ contains all the five basic relations.

It is easy to improve upon this by considering a larger tractable
constraint language than the set of basic relations.
Consider
the language
\[\Gamma' = \Theta^{\vee =} \setminus \{({\sf PP},{\sf PP}^{-1}), 
({\sf PP},{\sf PP}^{-1}, {\sf DR}),
({\sf PP},{\sf PP}^{-1}, {\sf EQ}),
({\sf PP},{\sf PP}^{-1}, {\sf DR}, {\sf EQ})\}\]
which is a well-known tractable class of RCC-5~\cite{DBLP:journals/ai/RenzN99}.
Assume that we extend the definition of a backdoor so that we allow $\alpha$ to be a function
from $B$ to $\Theta^{\vee =}$, i.e., one may assign a pair of
variables a union of the basic relations.
Then we can use $\Gamma'$ in connection with the simplification map
\[\Sigma(\{R(x,y)\})=\{ \{(R \cap \theta)(x,y)\} \; | \; \theta \in \{({\sf PP}),({\sf PP}^{-1},{\sf DR},{\sf PO},{\sf EQ}))\} \; {\rm and} \; R \cap \theta  \neq \emptyset.\}
\]
Note that $R \cap \theta$
with $\theta \in \{({\sf PP}),({\sf PP}^{-1},{\sf DR},{\sf PO},{\sf EQ}))\}$ is always a relation
in $\Gamma'$.
With this extension, we can perform backdoor solving with 
$\Gamma'$ and $\Sigma$ in $2^{|B|} \cdot \mathrm{poly}(||I||)$ time; the algorithm is a
straightforward generalisation of the algorithm underlying Theorem~\ref{thm:backdoor-eval}.

This idea can be generalised as follows. Instead of looking at subinstances containing two variables, we may
look at subinstances with a larger number of variables and perform a similar branching
process
where we replace subinstances containing ``problematic'' constraints with logically equivalent instances that
do not contain such constraints. This approach has the potential of sometimes leading to faster algorithms
than the backdoor approach (as indicated above and in the forthcoming Example~\ref{ex:sidedoor-2}).
We formalise and analyse the sidedoor approach in the following sections.

%********************
%
%
%{\bf Note:} Sioutis and Janhunen~\cite{Sioutis:Janhunen:ki2019} require refinements to be atomic---i.e. they are in the
%first case above and it leads to a slower algorithm.

\subsection{Definition of Sidedoors}
\label{sec:sidedoor-def}

We will now formally define sidedoors. We simplify the presentation by first
introducing {\em sidedoor triples} which should be viewed as the sidedoor analogue of
backdoor triples.

\begin{definition}
Let $\source$ and $\target$ denote two relational structures such that $\target \subseteq \source$, and let $r \geq 1$
be an integer.
We say that $\llbracket \source,\target,r \rrbracket$ is a {\em sidedoor triple} and we refer to
\begin{itemize}
\item
$\source$ as the {\em source language},

\item
$\target$ as the {\em target language}, and

\item
$r$ as the {\em radius}.
\end{itemize}
\end{definition}

In a backdoor triple $[\source,\target,\base]$, we require $\source$ and $\target$ to be qffo reducts of $\base$, so
$\source$ and $\target$ have an underlying connection via the structure $\base$.
Such a connection between $\source$ and $\target$ is not enforced in a sidedoor triple $\llbracket \source,\target,r \rrbracket$
where $\source$ and $\target$ can be chosen freely under the condition $\target \subseteq \source$.

%\todo{Is it worth contrasting this to backdoors (given that these assumptions are much milder than in the backdoor case)? DONE!}

\begin{definition} \label{def:sidedoor}
  Assume $\llbracket \source,\target,r \rrbracket$ is a sidedoor triple and
  let $(V,C)$ be an instance of CSP$(\source)$. We say that $S \subseteq 2^V$
  is a {\em sidedoor} into CSP$(\target)$ with radius $r$ if the following hold:

\begin{enumerate}
\item
each set in $S$ has size at most $r$, and

\item
for arbitrary $R(x_1,\dots,x_k) \in C$ with $R \not\in \target$, there is a set $s \in S$ such that
$\{x_1,\dots,x_k\} \subseteq s$.
\end{enumerate}
We say that a constraint $R(x_1,\dots,x_k) \in C$ is \emph{covered} by the sidedoor
$S$ if there is a set $s \in S$ such that $\{x_1,\dots,x_k\} \subseteq s$.
\end{definition}
The intuitive idea is that the constraints that are covered by the sidedoor are ``difficult'' to handle. 
Thus, one can concentrate on ``simplifying'' the constraints that are covered by the sidedoor: this can be done
via a branching strategy where the constraints that are covered by the sidedoor are replaced by computationally
tractable constraints in a solution-preserving way.
To make this work, one needs to restrict the size of the sets in the sidedoor: if we remove (1) from 
Definition~\ref{def:sidedoor}, then every instance $(V,C)$ of CSP$(\source)$ has the sidedoor $\{V\}$ of size 1!
Note that if a CSP$(\source)$ instance $(V,C)$ contains a constraint $R(x_1,\dots,x_k)$ where $R \not\in \target$ and where $x_1, \ldots, x_k$ are all distinct variables, then $(V,C)$
does not admit any sidedoor with radius $r < k$.
We remark that sidedoors work equally well for infinite and finite domains since we do not make any assumptions concerning
the domains of the languages $\source$ and $\target$.

We next introduce branching maps in order to formalise the idea of ``simplifying hard constraints''.

\begin{definition}
Assume that $\llbracket \source,\target,r \rrbracket$ is a sidedoor triple and let $x_1,\dots,x_r$ denote
distinct variables.
Let $\source_r$ be the set of CSP$(\source)$ instances
with variable set $\{x_1,\dots,x_r\}$ and
let $\target_r$ be the set of CSP$(\target)$ instances
with variable set $\{x_1,\dots,x_r\}$.
Let $\Omega$ be a mapping from $\source_r$ to $2^{\target_r}$ that satisfies the following condition: for every $I \in \source_r$,
\[\sols(I) = \bigcup_{I' \in \Omega(I)} \sols(I').\]
Then, we say that $\Omega$ is a {\em branching map} from $\source$ to $\target$ with radius $r$.
\end{definition}

Unlike simplification maps, it is essential that a branching map is a total function
from $\source_r$ to $2^{\target_r}$.
The {\em branching factor} of $\Omega$ is $\max\{|\Omega(I)| \; : \; I \in \source_r\}$, and it
is directly correlated to the complexity of the sidedoor evaluation problem as we will
see in Section~\ref{sec:sidedoor-eval}.
Note that if a a language $\Gamma$ contains infinitely many relations, then there may not be an upper bound on the size of CSP$(\Gamma)$ instances on $r$
variables. Thus, we sometimes need to restrict ourselves to polynomial-time computable
branching maps.
We now reconsider Examples~\ref{ex:languages} and~\ref{ex:rcc5} with respect to sidedoors.

%(*** Introduce polynomial bound ***)
%This implies that a branching map for an infinite language may output infinite sets for %certain inputs.
%Thus, we define the notion of bounded branching maps.
%Let $\Omega$ denote a branching map from $\source$ to $\target$ with radius $r$ and let 
%\[\alpha(\Omega) = \sup \left\{\langle \Omega(I) \rangle \; | \; \mbox{$I$ is a %CSP$(\source)$ instance on variables $\{x_1,\dots,x_r\}$} \right\} \]
%where $\langle \Omega(I) \rangle$ denote the total number of constraints appearing in the %instances in $\Omega(I)$.
%If $\alpha(\Omega) < \infty$, then we say that $\Omega$ is {\em bounded}.

%\begin{example} \label{ex:findom-sidedoor}
%Let $D=\{1,\dots,d\}$ and let $\source$ denote a finite constraint language
%with domain $D$. If the maximum arity of the relations in $\source$ is $a$, then
%there exists a simplification map $\Omega$ from $\source$ to $\target=\{\{(0)\},\dots,\{(d)\}\}$
%with radius $a$. This follows immediately from the fact that we can view each relation
%$R \subseteq D^k$ in $\source$ as
%
%\[R(x_1,\dots,x_k) \equiv \bigvee_{t \in R} x_1=t[1] \wedge x_2=t[2] \wedge \dots \wedge x_k=t[k].\]
%
%The map $\Omega$ has branching factor at most $d^a$.
%\end{example}

\begin{example} \label{ex:eqlang-sidedoor}
Recall that the equality relation $\delta$ from Example~\ref{ex:eq} is defined as follows:
\[\delta=\{(y_1,y_2,y_3) \in {\mathbb Q}^3 \; | \;
(y_1 = y_2 \land y_1 \neq y_3) \lor (y_1 \neq y_2 \land y_2 = y_3)\}.\]
Based on the definition of $\delta$, we see that the \NP-hard problem CSP$(\{\delta\})$ admits a simple branching map $\Omega$ from $\{\delta\}$
to $\{=,\neq\}$ with radius 3 and branching factor 2. Let $x_1,x_2,x_3$ denote distinct variables and let $V=\{x_1,x_2,x_3\}$.
We see, for instance, that
\[\Omega((V,\{\delta(x_1,x_2,x_3)\}))= \{ (V,\{x_1=x_2,x_1 \neq x_3\}),(V,\{x_1 \neq x_2,x_2=x_3\})\},\]
\[\Omega((V,\{\delta(x_1,x_1,x_3)\}))= \{ (V,\{x_1 \neq x_3\})\},\]
and
\[\Omega((V,\{\delta(x_1,x_2,x_3),\delta(x_2,x_3,x_1)\})) = \emptyset\]

%\todo{$\Omega$ should take a CSP instance of the form $(V,C)$ as input, not just a set of constraints. FIXED! (in both examples)}

%This branching map with branching factor 2.
%Currently best algorithm for CSP$(S)$ in $3^n$ time~\cite{Jonsson:Lagerkvist:ijcai2020}. %Sidedoor approach
%nice for small sidedoors.
\end{example}

\begin{example} \label{ex:sidedoor-1}
%Consider an instance $I=(V,C)$ of CSP$(\Theta^{\vee =})$ where
%
%\begin{itemize}
%\item
%$V=\{x_0,\dots,x_{n-1}\}$ where $n$ is divisible by three.
%
%\item
%the relations between variables $x_{3k},x_{3k+1},x_{3k+2}$ are in $\Gamma \setminus \Gamma'$ %for $0 \leq k \leq n/3-1$, and
%
%\item
%all other constraints are in $\Gamma'$.
%\end{itemize}
%
%We partition $V$ into $n/3$ partitions $V_k=\{x_{3k},x_{3k+1},x_{3k+2}\}$, $0 \leq k \leq %n/3-1$.
%It is easy to verify that $I$ has a backdoor 
%\[B=\{\{x_i,x_j\} \; | \; x_i,x_j \in V_k, \; x_i \neq x_j, \; {\rm and} \; 0 \leq k \leq %n/3-1\}\]
%with $|B|=n$. The set $B$ is additionally a sidedoor with radius 2. 

We know from Section~\ref{sec:motex} that CSP$(\Gamma')$ is polynomial-time solvable when
\[\Gamma' = \Theta^{\vee =} \setminus \{({\sf PP},{\sf PP}^{-1}), 
({\sf PP},{\sf PP}^{-1}, {\sf DR}),
({\sf PP},{\sf PP}^{-1}, {\sf EQ}),
({\sf PP},{\sf PP}^{-1}, {\sf DR}, {\sf EQ})\}.\]
We will use the relation $({\sf PP}^{-1},{\sf DR},{\sf PO},{\sf EQ})) \in \Gamma'$ 
frequently so we let $\lambda$ denote it.
Let $x_1,x_2,x_3$ denote distinct variables and let $V_2=\{x_1,x_2\}$ and $V_3=\{x_1,x_2,x_3\}$.

A suitable branching map from $\Theta^{\vee =}$ to $\Gamma'$
with radius 2 and branching factor 2 is the following:
\[\Omega_2((V_2,\{R(x_1,x_2)\}))=\{ (V_2,\{(R \cap ({\sf PP}))(x_1,x_2)\}),
                      (V_2,\{(R \cap \lambda)(x_1,x_2)\}) \}.\]
%The set
%$S=\{V_0,\dots,V_{n/3-1}\}$
%is a sidedoor with radius 3 and $|S|=n/3$,
To construct a branching map $\Omega_3$ from $\Theta$
to $\Gamma'$ with radius 3, we simply extend the idea behind $\Omega_2$. Thus,
consider the following branching map:
\[
\begin{split}
& \Omega((V_3,\{R_{12}(x_1,x_2),R_{23}(x_2,x_3),R_{13}(x_1,x_3)\}))=  \\
& \left\{ (V_3,\{(R_{12} \cap r_{12})(x_1,x_2),(R_{23} \cap r_{23})(x_2,x_3),(R_{13} \cap r_{13})(x_1,x_3) \}) \; | \; r_{12},r_{23},r_{13} \in \{({\sf PP}),\lambda\} \right\}
\end{split}
\]
It has branching factor 8 but this can be improved.
It is easy to verify that the constraint set $\{{\sf PP}(x_1,x_2), {\sf PP}(x_2,x_3), \lambda(x_1,x_3)\}$ is not satisfiable: note that
the constraints ${\sf PP}(x_1,x_2)$ and ${\sf PP}(x_2,x_3)$ forces the relation ${\sf PP}(x_1,x_3)$ to hold since ${\sf PP}$ is a transitive relation.
Thus, $\Omega$ can be refined into a branching map $\Omega_3$ with branching factor 7.
%It follows that the constraints ${\sf PP}(x,y)$ and ${\sf PP}(y,z)$ forces the relation ${\sf PP}(x,z)$ to hold.
%sidedoor approach based on $\Omega_3$ solves $I$ in time $O(7^{n/3}) \subseteq O(1.9130^n)$.
\end{example}

\subsection{Computing Branching Maps} 
\label{sec:branch-map}

Assume that $[\source,\target,\base]$ is a backdoor triple with
finite $\source$ and $\target$. Then, a simplification map from $\source$
to $\target$ always exists, and it can be computed whenever CSP$(\base)$ is decidable
by Lemma~\ref{lemma1-c}. If we consider a sidedoor triple 
$\llbracket \source, \target, r \rrbracket$, then a branching map does
not always exist. This is obvious if the arity of some relation
in $\source \setminus \target$ exceeds $r$, but there are other reasons
for the non-existence of a branching map, too. 
%\todo{Which reasons? FIXED!}
For instance, if we let $R_+ = \{(x,y,z) \in {\mathbb Q} \; | \; x=y+z\}$
and $R_{*} = \{(x,y,z) \in {\mathbb Q} \; | \; x=y \cdot z\}$, then it is easy to
verify that $\llbracket \{R_+,R_{*}\},\{R_+\},3 \rrbracket$ does not admit
a branching map and, consequently, $\llbracket \{R_+,R_{*}\},\{R_+\},r \rrbracket$
does not admit a branching map for any $r \geq 3$.

For every backdoor triple $[\source,\target,\base]$, we know that
$\source$ and $\target$ are qffo definable in $\base$.
Inspired by this, we show how to exploit qffo definability for identifying a natural class of
sidedoor triples that admit branching maps.
Recall that a first-order formula is {\em positive} if it does not contain negation.

\begin{lemma} \label{lem:branchmapcompute}
Let $\llbracket \source,\target,r \rrbracket$ be a sidedoor triple where $\source$ and $\target$ are finite.
Assume that the following holds:

\begin{enumerate}
\item
$\source$ is quantifier-free positively definable in $\target$, and

\item
the maximum arity $a$ of the relations in $\source$ satisfies $a \leq r$.
\end{enumerate}

Then, one can compute a branching map from $\source$ to $\target$ with radius $r$.
\end{lemma}
\begin{proof}
Let $x_1,\dots,x_r$ be distinct variables.
Let $\source_r$ and $\target_r$ denote the set of all CSP$(\source)$ and CSP$(\target)$ instances on variable set $\{x_1,\dots,x_r\}$, respectively. We construct
a branching map $\Omega$.
Arbitrarily pick $I=(V,C) \in \source_r$.

If every constraint $R(\bar{x})$ in $C$ satisfies $R \in \target$, then
we let $\Omega(I)=I$.
Otherwise,
pick one constraint $R(\bar{x}) \in C$ such that $R$
is not in $\target$. We may
(without loss of generality) assume that the positive qffo $\target$-definition 
of $R$ is in DNF since the 
conversion to DNF can be done without
introducing any negations. 
%\todo{In the lemma we assume  a quantifier-free positive definition, not a quantifier-free first-order definition, so this reads a bit strange. FIXED!}
Hence, $\phi(\bar{x})=\psi_1(\bar{x}) \vee \dots \vee \psi_m(\bar{x})$
is the DNF definition of $R(\bar{x})$ in $\target$
where
each $\psi_i$ is a conjunction of constraints based on (unnegated) relations in $\target$.
Obviously, each $\psi_i$ can
be viewed as an instance of CSP$(\target)$.
Let 
\[X = \{(V,(C \setminus \{R(\bar{x})\}) \cup \psi_1(\bar{x})), \dots,
(V,(C \setminus \{R(\bar{x})\}) \cup \psi_m(\bar{x}))\}.\]
It is an immediate consequence of $\phi(\bar{x})$ being a definition of $R(\bar{x})$ 
that
\[\sols(I) = \bigcup_{I' \in X} \sols(I')\]
If the instances in $X$ now contain constraints with relations from $\target$ only, then we are done and we let $\Omega(I)=X$.
Otherwise,
repeat the process of recursively replacing relations in $X$ that are not in $\Gamma'$ by their definitions in $\Gamma'$.
At least one constraint $R(\bar{x}) \in C$ with $R \in \source \setminus \target$ is removed
in every recursive step so we will eventually end up
with a non-empty set of instances $X'$ satisfying $\sols(I) = \bigcup_{I' \in X'} \sols(I')$ and where
each member of $X'$ is an instance of CSP$(\target)$ (since $r \geq a$).
Thus, we let $\Omega(I) = X'$. 
\end{proof}

\begin{example}
    It is easy to find examples where Lemma~\ref{lem:branchmapcompute} is applicable: if $\target$ is homogeneous and finitely bounded (see Example~\ref{ex:finbound}) and $\source$ a first-order reduct of $\target$ then it can also be defined as a qffo-reduct of $\target$.
\end{example}

The branching map computed in Lemma~\ref{lem:branchmapcompute} has only a finite number of possible inputs and can consequently be accessed in constant time.
It is obviously not guaranteed to have minimal branching factor---in particular, note
that the branching factor may vary depending on the choice of DNF definitions for
the relations in $\source$.

We emphasise that
the restriction to quantifier-free definitions in Lemma~\ref{lem:branchmapcompute} is necessary.
Define $A_k(x,y) \equiv x=y+k$. We see that
$A_2(x,y)=\exists z. A_1(x,z) \wedge A_1(z,y)$ so $A_2$ is fo-definable in $\{A_1\}$.
Let $\target=\{A_1\}$ and $\source=\{A_1,A_2\}$. 
There is no branching map from $\source$ to $\target$ with radius 2: this would imply that
$A_2(x,y)$ is logically equivalent to either $A_1(x,y)$ or $A_1(x,y) \vee A_1(y,x)$.

It is also necessary that the definitions are positive. Consider the equality and disequality relations $=$ and $\neq$
over some domain $D$.
Define the source language $\source=\{=,\neq\}$ and
the target language $\target=\{=\}$.
Obviously, every relation in $\source$ is qffo definable in $\target$ (but
$\neq$ is not qffo positive definable in $\target$).
We see that there is no branching map of radius $2$ from $\source$ to $\target$:
this would imply that $\neq$ equals $=$.
However, if the target language is finite and JEPD, then the result holds even if we use general quantifier-free
definitions: every
atomic formula $\neg R(x,y)$ can be replaced by
$\bigvee_{S \in \Gamma' \setminus \{R\}} S(x,y)$ as was pointed out
in Section~\ref{sec:simp-map}.
This implies the following connection between backdoors and sidedoors.

\begin{lemma}
Assume that $[\source,\target,\base]$ is a backdoor triple with
finite $\source,\target,\base$ and that $a$ is the maximum arity of the relations
in $\source$. Then, there is a branching map
for the sidedoor triple $\llbracket \source \cup \base, \target \cup \base, r \rrbracket$ whenever $r \geq \max(2,a)$.
\end{lemma}
\begin{proof}
This follows immediately from Lemma~\ref{lem:branchmapcompute} since $\source \cup \base$ is qffo positive definable in
the JEPD language $\base$ which is a subset of $\target \cup \base$.
\end{proof}

We require that $r \geq \max(2,a)$ since the relations in $\base$ have arity 2.
This way we cover the (fairly trivial) case when $\source$ only contains relations of arity 1.

%\medskip
%
%{\bf Questions.} Is there any use for branching maps where the branching factor is not minimal? 
%
%\bigskip
%
%Assume that $\Gamma,\Gamma'$ satisfy the preconditions of Lemma~\ref{lem:branchmapexists} and that
%the largest arity of any relation in $\Gamma$ is $a$. Then, every instance $(V,C)$ of CSP$(\Gamma)$
%admits a sidedoor into CSP$(\Gamma')$ of radius $r \geq a$.

\subsection{Sidedoor Evaluation}
\label{sec:sidedoor-eval}

We begin this section by introducing the sidedoor evaluation problem.

\pbDef{\textsc{$\llbracket \source,\target,r \rrbracket$-Sidedoor Evaluation}}{An instance $(V,C)$
  of CSP$(\source)$ and a sidedoor $S \subseteq 2^{V}$ of radius $r$ into CSP$(\target)$.}
  {Is $(V,C)$ satisfiable?}
  
We will now analyse the complexity of \textsc{$\llbracket \source, \target,r \rrbracket$-sidedoor evaluation}. We show that this problem is fixed-parameter tractable
(under natural side conditions)
when parameterised by sidedoor size.
The additional conditions concern the existence of a branching map that can be accessed in polynomial time and that
has bounded branching factor. 
We need some notation for describing our algorithm for this problem.
%Let $\Gamma \subseteq \Gamma'$ denote constraint languages, $r$ an integer, and $\Omega$ a branching map
%from $\Gamma$ to $\Gamma'$ with radius $r$.
%Let ${\mathcal I}$ be the set of CSP$(\Gamma)$ instances
%with at most $r$ variables and
%let ${\mathcal I}'$ be the set of CSP$(\Gamma')$ instances
%with at most $r$ variables.
Given a CSP instance $I=(V,C)$ and a set $V' \subseteq V$, we define the {\em restriction} of 
$I$ to $V'$ (denoted by $I[V']$) as the instance $(V',\{R(x_1,\dots,x_k) \in C \; | \; \{x_1,\dots,x_k\} \subseteq V'\})$.

Let $I=(V,C)$ be an arbitrary instance of CSP$(\Gamma)$
and $I'=(V',C')$ be an arbitrary instance of CSP$(\Gamma')$
with $V' \subseteq V$. We let $I \oplus I'$ denote the
instance 
\[(V, (C \setminus \{R(x_1,\dots,x_k) \in C \; | \; \mbox{$\{x_1,\dots,x_k\} \subseteq V'$}\}) \cup I').\]
That is, $I \oplus I'$ is the instance $I$ where every constraint that is completely covered by $V'$
is removed, and then the constraints in $I'$ are added to $I$.
Clearly, if $I[V']$ and $I'$ have the same set of satisfying assignments, then
$I$ is satisfiable if and only if $I \oplus I'$ is satisfiable.

\begin{figure}

{\bf algorithm} $A((V,C),S)$
\begin{enumerate}
\item {\bf if} $S=\emptyset$ {\bf then return} $A'((V,C))$ 
\item arbitrarily choose $s \in S$

\item let $\{I'_1,\dots,I'_p\} = \Omega'((V,C)[s])$

\item apply algorithm $A$ to the instances

  \begin{itemize}
\item $J_1=((V,C) \oplus I_1, S \setminus \{s\})$

\item $J_2=((V,C) \oplus I_2, S \setminus \{s\})$

\item $\vdots$

\item $J_p=((V,C) \oplus I_p, S \setminus \{s\})$
\end{itemize}

\item {\bf if} any call returns \textsc{Yes} {\bf then return} \textsc{Yes} {\bf else return} \textsc{No}
\end{enumerate}

\caption{Algorithm for sidedoor evaluation.}
\label{fig:sidedooralg}
\end{figure}

\begin{theorem} \label{thm:sidedoor-eval}
Assume that $\llbracket \source,\target,r \rrbracket$ is a sidedoor triple.
 Let $\Omega$ denote a polynomial-time computable branching map from $\source$ to $\target$ with radius $r$
 and branching factor $c < \infty$. If CSP$(\target)$ is polynomial-time solvable, then
\textsc{$\llbracket \source, \target,r \rrbracket$-sidedoor evaluation} is solvable in time
$c^{|S|} \cdot \mathrm{poly}(||I||)$ for a given sidedoor $S$. In particular, \textsc{$\llbracket \source, \target,r \rrbracket$-sidedoor evaluation} is fpt when parameterised by $|S|$.

%Let $\Gamma' \subseteq \Gamma$ be two finite constraint languages and let $\Sigma$ be a branching map from $\Gamma$ to $\Gamma'$
%with branching factor $c$. Let $I=(V,C)$ denote an instance of CSP$(\Gamma)$ that has a sidedoor $B$.
%Then algorithm $A$ solves $I$ in time 
%$s^{|B|} \cdot T'(||I||) \cdot poly(||I||)$ time where $s$ is the branching factor of $\Sigma$ and
%$T'$ is the time complexity of algorithm $A'$.
\end{theorem}
\begin{proof}
We assume that $A'$ is a polynomial-time algorithm for CSP$(\target)$.
Let $(V,C)$ be an instance of CSP$(\source)$ that has a sidedoor $S$ into CSP$(\target)$ with radius $r$.
We assume (without loss of generality) that every $s \in S$ satisfies $|s|=r$. If this is not the case, then
those sets that are smaller than $r$ can be augmented with arbitrary variables.
The branching map $\Omega$ is only defined for constraints over variable sets $X=\{x_1,\dots,x_r\}$.
We modify it into a map $\Omega'$ that is applicable to arbitrary sets of variables as follows. 
Let $Y=\{y_1,\dots,y_r\}$ denote a set of variables and let $\sigma$ denote an arbitrary
bijection from $X$ to $Y$. Given an instance $(Y,C_Y)$ of CSP$(\source)$, we let $\sigma((Y,C_Y))=(X,C_X)$ denote the instance
$(Y,C_Y)$ with variables $y_i$, $1 \leq i \leq r$, replaced by $\sigma(y_i)$. Let $\{I_1,\dots,I_p\}=\Omega((X,C_X))$
and $\Omega'((Y,C_Y))=\{\sigma^{-1}(I_1),\dots,\sigma^{-1}(I_p)\}$.

We claim that $(V,C)$
can be solved by Algorithm $A$ in Figure~\ref{fig:sidedooralg}. 
We first show that the algorithm always computes the correct answer.
To prove this, we begin by verifying that $S \setminus \{s\}$ is a sidedoor for $(V,C) \oplus I_i$, 
$1 \leq i \leq p$.
Arbitrarily choose one of the instances $J_q=((V,C_q),S \setminus \{s\})$, $1 \leq q \leq p$, from line (4) where
$C_q = C \setminus \{R(x_1,\dots,x_k) \in C \; | \; \mbox{$\{x_1,\dots,x_k\} \subseteq B'$}\} \cup C_q$ by the
definition of $(V,C) \oplus I_q$.
We verify the two properties in Definition~\ref{def:sidedoor}.

\begin{enumerate}
\item
The set $S$ is a sidedoor for $(V,C)$ with radius $r$ so each set in $S$ has size $r$.
The same obviously holds for $S \setminus \{s\}$.

\item
We need to show that
for arbitrary $R(x_1,\dots,x_k) \in C_q$ with $R \not\in \target$, there is a set $s' \in S \setminus \{s\}$ such that
$\{x_1,\dots,x_k\} \subseteq s'$. Arbitrarily choose $R(x_1,\dots,x_k) \in C_q$ with $R \not\in \target$.
Assume to the contrary that there is no $s' \in S \setminus \{s\}$ such that
$\{x_1,\dots,x_k\} \subseteq s'$. Since $S$ is a sidedoor for $I$, this implies that
$\{x_1,\dots,x_k\} \subseteq s$. This leads to a contradiction since all relations in
$\Omega'((V,C)[s])$ are in $\target$ by the definition of branching maps.

\end{enumerate}

%\medskip

We prove correctness by induction over $|S|$.
If $|S|=0$, then $(V,C)$ is an instance of CSP$(\target)$ and the correct answer is computed in line (1).
Assume the algorithm computes correct answers whenever $|S| < m$. We show that this holds also when $|S|=m$.
Consider the CSP instances

\medskip

\hspace{2em} $J_1'=(V,C) \oplus I_1$
 
 \hspace{2em} $J_2'=(V,C) \oplus I_2$

 \hspace{2em} $\vdots$

\hspace{2em} $J_p'=(V,C) \oplus I_p$

\medskip

\noindent
that are computed in line (4).
Condition (2) in the definition of branching maps immediately implies that
%\[\{f:V \rightarrow D \; | \; \mbox{$f$ is a solution to $I$}\} = \bigcup_{i=1}^p \{f:V \rightarrow D \; | \; \mbox{$f$ is a %solution to $J_i'$}\}.\]
\[\sols(I) = \bigcup_{i=1}^p \sols(J_i').\]
If the algorithm answers \textsc{Yes}, then
\[A((V,C) \oplus I_q,S \setminus \{s\})=\mbox{\textsc{Yes}}\]
for some $1 \leq q \leq p$.
The induction hypothesis implies that
$(V,C) \oplus I_q)$ is satisfiable
since $|S \setminus \{s\}| < m$ and we know that
$S \setminus \{s\}$ is a sidedoor for
$(V,C) \oplus I_q$.
Hence, $\sols((V,C))$ is non-empty and $(V,C)$ has a solution.
Similarly, if the algorithm answers \textsc{No}, then none of $J_1',\dots,J_p'$ has a solution so the set
$\sols((V,C))$ is empty and
$(V,C)$ has no solution.

\medskip

We conclude the proof by analysing the time complexity of algorithm $A$.
Let $s$ denote the size of the sidedoor.
%The constraint language $\Gamma$ is finite so any instance $(V,C)$ of CSP$(\Gamma)$ (and, consequently, any instance
%$(V,C)$ of CSP$(\Gamma')$)
%contains at most $O(|V|^a)$ constraints where $a$ is the maximal
%arity of the relations in $\Gamma$. 
%Note that $V$, and thus $|V|$, is never changed during
%an execution of algorithm $A$.
Now consider a branch in the algorithm
\[(V,C) \rightarrow (V,C_1) \rightarrow \dots \rightarrow (V,C_s)\]
where $(V,C)$ is the original instance, $(V,C_1)$ is an instance $(V,C) \oplus I_i$, and so on.

Since $\Omega$ (and consequently $\Omega'$) is polynomial-time computable, 
we can assume that $||\Omega(I)|| \leq t(||I||)$ for some polynomial $t$.
This implies that
\[||(V,C) \oplus I_i || \leq ||(V,C)||+t(||(V,C)[s]||) \leq ||(V,C)||+t(||(V,C)||).\]
Hence,
we can without loss of generality assume that
$(V,C_p)$, $1 \leq p \leq s$, have size at most $||(V,C)||+ p \cdot t(||(V,C)||)$. Let $b=||(V,C)|| \cdot s \cdot t(||(V,C)||)$ and note that
$b$ is trivially an upper bound on $|C|+p \cdot t(||(V,C)||)$, $1 \leq p \leq s$.
We simplify our argument by assuming that all CSP instances encountered during
an execution of the algorithm 
has size $b$.

By inspecting algorithm $A$, we see that its running time is bounded by a function $T$ in the size of the sidedoor
with recursive definition
%\[T(0) = T'(\beta)\]
%\[T(s) = c \cdot T(s-1) + poly(\beta), \; s > 0\]
\[ T(s)= \left\{
\begin{array}{ll}
T'(b) & \mbox{if $s=0$} \\
c \cdot T(s-1) + {\rm poly}(b) & \mbox{if $s > 0$}
\end{array}
\right.\]
where $c$ is the branching factor of $\Omega$ and $\Omega'$, and $T'$ is the time
complexity of algorithm $A'$.
Thus, for sufficiently large $||I||$ there exists constants $q,q'$ such that
\[T(s) \leq c^s \cdot T'(b) \cdot \mathrm{poly}(b)
\leq c^s \cdot b^q \cdot b^{q'} = c^s \cdot (||I|| \cdot s \cdot \alpha)^{q+q'}
= c^s \cdot ||I||^{q+q'} \cdot (s \cdot \alpha)^{q+q'}.\]

We conclude that the algorithm is fixed-parameter tractable since $c$ and $\alpha$
are constants that do not depend on the input.
\end{proof}

By applying Theorem~\ref{thm:sidedoor-eval} to Example~\ref{ex:eqlang-sidedoor}, we see that
CSP$(\{S\})$ can be solved in $O(2^{s})$ time where $s$ is sidedoor size.
We illustrate some more aspects of sidedoors in the next example.

\begin{example} \label{ex:sidedoor-2}
Let us once again consider RCC-5 and Example~\ref{ex:sidedoor-1}.
Assume $I$ is an instance of CSP$(\Theta^{\vee =})$ with $n$ variables and that it has an $s$-element sidedoor with radius $r$.
Theorem~\ref{thm:sidedoor-eval} implies that
$I$ can be solved in $O(2^s \cdot {\rm poly}(||I||))$ time by using branching map $\Omega_2$ if $r=2$ and in
$O(3^s \cdot {\rm poly}(||I||))$ time by using branching map $\Omega_3$ if $r=3$.
This represents an improvement for instances with small sidedoors since the fastest 
known algorithm for CSP$(\Theta^{\vee =})$ runs in
$2^{O(n \log n)}$ time~\cite{Jonsson:etal:ai2021}. More precisely, the sidedoor approach beats this algorithm for instances 
with sidedoor size in $o(n \log n)$.

We now illustrate how sidedoors with large
radius may be beneficial. Let ${\mathcal I}$ denote the set of
RCC-5 instances $I=(V,C)$ where

\begin{itemize}
\item
$V=\{x_0,\dots,x_{n-1}\}$ where $n$ is divisible by three.

\item
the relations between variables $x_{3k},x_{3k+1},x_{3k+2}$ are in $\Gamma \setminus \Gamma'$ for $0 \leq k \leq n/3-1$, and

\item
all other constraints are in $\Gamma'$.
\end{itemize}

We partition $V$ into $n/3$ partitions $V_k=\{x_{3k},x_{3k+1},x_{3k+2}\}$, $0 \leq k \leq n/3-1$.
Arbitrarily choose $I \in {\mathcal I}$.
We know that $I$ has a backdoor 
\[B=\{\{x_i,x_j\} \; | \; x_i,x_j \in V_k, \; x_i \neq x_j, \; {\rm and} \; 0 \leq k \leq n/3-1\}\]
with $|B|=n$ and this set
is additionally a sidedoor with radius 2.
If we use the branching map $\Omega_2$ from Example~\ref{ex:sidedoor-1} together with the sidedoor $B$, then 
the CSP for the instances in ${\mathcal I}$ can be solved in $O(2^n \cdot {\rm poly}(||I||))$ time by
Theorem~\ref{thm:sidedoor-eval} since $\Omega_2$ has branching factor 2. Note that the sidedoor approach beats the backdoor approach (using backdoor $B$) which runs in $O(4^{|B|} \cdot {\rm poly}(||I||) = O(4^{n} \cdot {\rm poly}(||I||)$ time
according to Example~\ref{ex:rcc5}.

Now,
it is easy to verify that the set
$S=\{V_0,\dots,V_{n/3-1}\}$
is a sidedoor with radius 3 and $|S|=n/3$.
This implies (by Theorem~\ref{thm:sidedoor-eval}) that the CSP for the instances in ${\mathcal I}$ can be solved 
in $O(7^{n/3} \cdot {\rm poly}(||I||)) \subseteq O(1.9130^n \cdot {\rm poly}(||I||))$ time 
since $\Omega_3$ has branching factor 7. 
\end{example}

\subsection{Sidedoor Detection}
\label{sec:sidedoor-detection}

We begin by introducing the sidedoor detection problem.

\pbDef{\textsc{$\llbracket \source, \target,r \rrbracket$-Sidedoor Detection}}{An instance $(V,C)$
  of CSP$(\source)$ and an integer $k$}
  {Does $(V,C)$ contain a sidedoor to $\target$ of radius $r$ and size at most $k$.}

If $\source,\target$ contain only unary relations, then 
{\sc Sidedoor Detection} is a trivial problem (regardless of the radius). 
Similarly, if $\source,\target$ contain relations that are at most binary, then
$\llbracket \source,\target,2 \rrbracket$-{\sc Sidedoor Detection} is a trivial problem.
Thus, the most basic case when {\sc Sidedoor Detection} is possibly
computationally hard is when $\source,\target$ contain at most
binary relations and the radius is 3.
We prove \NP-hardness of this case via a reduction from a graph partitioning problem.
A {\em partition} of a set $S$ is a collection of disjoint subsets of $S$ such that
their union equals $S$. Let $K_n$ denote the undirected complete graph on $n$ vertices, i.e.
every two distinct vertices are connected by an edge.

\pbDef{\textsc{Edge Partition}$(r)$}{Undirected graph $G=(V,E)$}
  {Can $E$ be partitioned into sets $E_1,E_2,\dots$ such that each
$E_i$ induces a subgraph of $G$ that is isomorphic to $K_r$?}

\begin{theorem} (\cite{Holyer:sicomp81})
{\sc Edge Partition}$(r)$ is \NP-complete for every $r \geq 3$.
\end{theorem}

We can now demonstrate that $(\source,\target,r)$-{\sc Sidedoor Detection}
is in general an \NP-hard problem.

\begin{theorem}
$\llbracket \{R,R'\}, \{R\},3 \rrbracket$-{\sc Sidedoor Detection} is \NP-hard when
$R,R'$ are arbitrary binary relations.
\end{theorem}
\begin{proof}
We prove this by a polynomial-time reduction from {\sc Edge Partition}$(3)$.
%Let $\source=\{R,R'\}$ and $\target=\{R'\}$.
Arbitrarily choose an undirected graph $G=(V,E)$ and assume without loss of generality
that $|E|/3$ is an integer.
Construct an instance $I=((V,C),k)$ of $\llbracket \{R,R'\},\{R\},r \rrbracket$-{\sc Sidedoor Detection} where $C=\{R(x,y) \; | \; \{x,y\} \in E\}$
and $k=|E|/3$.

If $G$ is a \textsc{Yes}-instance of {\sc Edge Partition}$(3)$, then let $E_1,\dots,E_m$ denote such a partition of the edges in $G$.
Each set $E_i$, $1 \leq i \leq m$, contains three edges so $m=|E|/3$.
Each $E_i$, $1 \leq i \leq m$, is defined on a set of three vertices so we let $S_i$, $1 \leq i \leq m$, denote the corresponding
set of vertices. It is obvious that $\{S_1,\dots,S_m\}$ is a sidedoor of radius 3 for $(V,C)$ and it contains at most
$|E|/3$ sets. Thus $I$ is a \textsc{Yes}-instance.

Assume now that $I$ is a \textsc{Yes}-instance of $\llbracket \{R,R'\},\{R\},3 \rrbracket$-{\sc Sidedoor Detection}. 
Then there is a sidedoor
$\{S_1,\dots,S_k\}$ with $k=|E|/3$ and each $S_i$ contains at most three variables since the radius is three.
The instance $(V,C)|_{S_i}$, $1 \leq i \leq k$, must thus contain three distinct constraints since
the sidedoor cover all constraints in $C$ and $|C|=|E|$. Furthermore, if there is a constraint
$R(x,y)$ in $C$ such that $\{x,y\} \subseteq S_i$ and $\{x,y\} \subseteq S_j$ for some
$1 \leq i \neq j \leq k$, then the sidedoor covers strictly less that $|E|$ constraints in $C$.
This implies that the sets in the sidedoor are pairwise disjoint.
We conclude that the sidedoor is a partitioning of the set $V$ and that $E_1,\dots,E_k$ with
\[E_i=\{\{v,w\} \; | \; \mbox{$v,w$ are distinct elements in $S_i$}\}\] 
is an edge partition of $G$. The graph $G$ is consequently a \textsc{Yes}-instance.
\end{proof}

We next prove that $\llbracket \source,\target,r \rrbracket$-{\sc Sidedoor Detection} is fpt when parameterised by 
the size of the sidedoor.
We say that a variable $v \in V$ in a CSP instance $(V,C)$ is {\em isolated} if it does not appear in the
scope of any constraint in $C$.

\medskip

\begin{theorem}
\label{thm:sidedoor-detection}
$\llbracket \source,\target,r \rrbracket$-{\sc Sidedoor Detection} is solvable in $(rk)^{(rk)} \cdot \mathrm{poly}(||I||)$ time where $k$
is the size of the sidedoor. Hence, the problem is in $\FPT$ when parameterised by $k$.
\end{theorem}
\begin{proof}
Let $((V,C),k)$ be an instance of $\llbracket \source,\target,r \rrbracket$-{\sc Sidedoor Detection}.
Let $(V,C')$ be $(V,C)$ where all constraints whose relation is in $\target$ are removed.
Let $(V',C')$ be $(V,C')$ with all isolated variables removed.
If $|V'| > rk$, then $(V,C)$ does not have a sidedoor of size at most $k$
since every variable in $V'$ must be contained in at least one set of the sidedoor, and a sidedoor
of radius $r$ with $k$ elements can cover at most $rk$ variables.
Otherwise, enumerate all sets $\{S_1,\dots,S_k\}$ where $S_i$, $1 \leq i \leq k$ is a
subset of $V'$ of size $r$. The instance $((V,C),k)$ has a sidedoor if and only if at least
one of these sets is a sidedoor. We see that there are at most
$((rk)^r)^k = (rk)^{rk}$ such sets: our universe has size at most $rk$, there are at most $(rk)^r$ subsets of size
$r$, and we want to choose at most $k$ such sets. 
Thus, the algorithm runs in $(rk)^{(rk)} \cdot \mathrm{poly}(||I||)$ time.
\end{proof}

%\todo{Any ideas for a better algorithm? Or an interesting lower bound?}

\section{Summary and Research Questions}
\label{sec:concludingremarks}

We have generalised the backdoor concept to CSPs over infinite domains and we have
presented parameterized complexity results for infinite-domain backdoors. Interestingly, despite being a strict generalisation of finite-domain backdoors, both backdoor detection and evaluation turned out to be in $\FPT$. Hence, the backdoor paradigm is applicable to infinite-domain CSPs, and, importantly, it is indeed possible to have a uniform backdoor definition (rather than having different definitions for equality languages, temporal languages, RCC-5, and so on).
We have noted that
the backdoor approach sometimes leads to inferior algorithms for binary constraints.
Inspired by this, we introduced the sidedoor approach where subinstances induced by the
sidedoor are rewritten in a solution-preserving way. We prove that the detection and evaluation problems for
sidedoors are in FPT, and we demonstrate that the time complexity of solving CSPs with sidedoors is
in general incomparable with the backdoor approach.
The running times of the two approaches are summarised in Table~\ref{tb:summary} where $|B|$ and $|S|$ denote backdoor and sidedoor size, respectively.
For a backdoor triple $[\source,\target,\base]$, we let $a$ denote the maximum arity of the relations in $\target$ and for
a sidedoor triple $\llbracket \source, \target, r \rrbracket$, we let $c$ denote the branching factor of the branching map
under consideration. The results for backdoors and sidedoors are, naturally, not directly comparable
since their sizes are not connected to each other in some obvious way.
One may still draw interesting conclusions, though, such that the complexity of sidedoor evaluation
is independent of the sidedoor triple---all that matters is the branching factor of the branching map. 
Thus, for problems admitting a branching map with low
branching factor, the sidedoor approach may be faster than the backdoor approach even
if sidedoors happen to be moderately larger than backdoors.

We continue this section by discussing a few possible directions for future research.

% \begin{table}  
%   \begin{tabular}{|l|l|l|l|l|}
%     \hline
%     \multirow{2}{*}{} &
%       \multicolumn{2}{c|}{Finite languages} &
%       \multicolumn{2}{c|}{Infinite languages} 
%       \\
%     & Backdoor & Sidedoor & Backdoor & Sidedoor  \\
%     \hline

%     Detection & $\binom{a}{2}^{|B|+1} \cdot |\base|^{|B|} \cdot {\rm poly}(||I||)$ &  $(r|S|)^{r|S|} \cdot {\rm poly}(||I||)$ & \Weft[2]-hard &  $(r|S|)^{r|S|} \cdot {\rm poly}(||I||)$  \\ \hline
%         Evaluation & $|\base|^{|B|} \cdot {\rm poly}(||I||)$ & $c^{|S|} \cdot {\rm poly}(||I||)$ & (*) & $c^{|S|} \cdot {\rm poly}(||I||)$  \\
%     \hline
%   \end{tabular}
 
%   \caption{Summary of complexity results. (*) indicates that this problem is not interesting since backdoor detection is computationally hard.}
%   \label{tb:summary}
% \end{table}

\begin{table}
  \centering
  \begin{tabular}{|l|l|l|l|l|}
    \hline
    \multirow{2}{*}{} &
      \multicolumn{2}{c|}{Finite languages} &
      \multicolumn{2}{c|}{Infinite languages} 
      \\
    & Backdoor & Sidedoor & Backdoor & Sidedoor  \\
    \hline

    Detection & $\binom{a}{2}^{|B|+1} \cdot |\base|^{|B|}$ &  $(r|S|)^{r|S|}$ & \Weft[2]-hard &  $(r|S|)^{r|S|}$  \\ \hline
        Evaluation & $|\base|^{|B|}$ & $c^{|S|}$ & (*) & $c^{|S|}$  \\
    \hline
  \end{tabular}
 
  \caption{Summary of complexity results (polynomial-time factors are omitted). (*) indicates that this problem is not interesting since backdoor detection is computationally hard.}
  \label{tb:summary}
\end{table}

\paragraph*{Broader tractable classes}

%Furthermore, the detection of heterogeneous backdoors is still fixed-parameter
%tractable in many natural cases.

Recent advances concerning backdoor sets for SAT and finite-domain CSPs provide a number of promising 
directions for future work. For instance, Gaspers et al. (2017a)\nocite{GASPERS201738} have introduced the idea of so-called {\em hetereogenous} backdoor sets, i.e. backdoor sets into the disjoint union of more than one base language, and Ganian et al. (2017) \nocite{GanianRS17} have exploited the idea that if
variables in the backdoor set separate the instance into several independent components, then the instance can 
still be solved efficiently as long
as each component is in some tractable base class. The detection of heterogeneous backdoors is still fixed-parameter
tractable in many natural cases. Both of these approaches significantly enhance the power and/or generality of the backdoor approach for
finite-domain CSP and there is a good chance that these concepts can also be lifted to infinite-domain CSPs.
These approaches are highly interesting for sidedoors, too. 

Another interesting direction that is relevant for both back- and sidedoors is to drop the requirement that they move the
instance to a polynomial-time
solvable class---it may be sufficient that the class is solvable in,
say, single-exponential
$2^{O(n)}$ time. This can lead to substantial speedups
when considering
CSPs that are not solvable in $2^{O(n)}$ time. Natural classes of this kind are known to exist under
the exponential-time hypothesis~\cite{Jonsson:Lagerkvist:mfcs2018}, and concrete examples are given by certain extensions of Allen’s algebra that are
not solvable in $2^{o(n \log n)}$ time. 
%Note that this idea is not interesting for finite-domain CSPs since whenever
%a constraint language $\Gamma$ has a finite domain $D$, then CSP$(\Gamma)$ is solvable in $2^{O(n)}$ time.

Another promising direction
for future research is compact representations of backdoors. This approach has been studied for finite-domain CSPs with the
aid of decision trees or the more general concept of {\em backdoor DNFs}~\cite{OrdyniakSS21,Samer:Szeider:aaai2008}. The idea is to use decision trees or backdoor DNFs for representing all (partial) assignments
of the variables in the backdoor set. This can lead to a much more efficient algorithm for backdoor evaluation since instead of considering all assignments of the backdoor variables, one only needs to consider a potentially much smaller set of partial assignments of those variables that (1) cover all possible assignments and (2) for each partial assignment the reduced instance is in the base class. It has been shown that this approach may lead to an exponential improvement of the backdoor evaluation problem in certain cases, 
and it has been verified experimentally that these kinds of 
backdoors may be substantially smaller than the standard ones~\cite{OrdyniakSS21,Samer:Szeider:aaai2008}.

\paragraph*{The detection and evaluation problems for infinite languages}
Our results show that 
there is a significant difference between problems based on finite constraint languages
and those that are based on infinite languages.
The backdoor detection and evaluation problems are fixed-parameter tractable when the languages
are finite.
In the case of infinite languages, we know that the backdoor detection problem is \Weft[2]-hard
for certain choices of languages. This raises the following question: for which infinite source
languages is backdoor detection fixed-parameter tractable? This question is probably very hard
to answer in its full generality so it needs to be narrowed down in a suitable way.
A possible approach is to begin by studying this problem for equality languages. 

We note that the situation is
slightly different for sidedoors. Given a sidedoor triple $\llbracket \source,\target,r \rrbracket$,
we know that the sidedoor detection problem is fpt even if $\source$ and/or $\target$ are infinite.
The problem with sidedoors is instead that the radius $r$ limits the set of instances that the method
is applicable to. A motivated question here is what happens if we let the radius be a function
of the instance instead of being a fixed constant. If the radius is allowed
to increase linearly in the number of variables, then it is possible to adapt the \Weft[2]-hardness
result in Section~\ref{sec:inf} to the sidedoor setting. 
If the radius grows sublinearly, then the consequences are not clear.
It is, however, clear that the algorithm underlying Theorem~\ref{thm:sidedoor-detection} cannot lead to fpt algorithms
when the radius is allowed to depend on the given instance, so other algorithmic techniques are needed for this approach to work.

\paragraph*{Computation of simplification maps}
We have presented an
algorithm for constructing simplification maps that works under the condition that
the source and target languages are finite. However,
we have no general method to compute simplification maps for infinite languages.
It seems conceivable that the computation of simplification maps is an undecidable
problem and proving this is an interesting research direction. 
From a general point of view, the problem which we want to address is
the following. Given a relation represented by a first-order formula over $\base$ (i.e., corresponding to a simplified constraint) we wish to decide whether it is possible to find a CSP$(\target)$ instance whose set of models coincides with this relation.
This problem is known in the literature as the {\em inverse constraint satisfaction problem} over a constraint language $\target$ (Inv-CSP$(\target,\base)$), and it may be defined as follows.
 
 \pbDef{Inv-CSP$(\target, \base)$}{A relation $R$ (represented by an fo-formula over $\base$).}
 {Can $R$ be defined as the set of models of a CSP$(\target)$ instance?}

We are thus interested in finding polynomial-time solvable cases of this problem, since this would imply the existence of an efficiently computable simplification map to $\target$ even if the source language is infinite.
%where one given a relation, wishes to know whether there exists a CSP$(\target)$ instance whose set of models coincide with the given relation.
The Inv-CSP problem has been fully classified for the Boolean domain~\cite{kavvadias98,lagerkvist2020e}, but little is known for arbitrary finite domains, and even less has been established for the infinite case. 
%Since we represent constraint languages via reducts of fixed relational structures, one 
%may want to restrict the given relation to be fo-definable by $\base$. 
%Thus, the input relation $R$ corresponds to a (potentially simplified) constraint, and we wish to know whether this relation can be defined as a CSP$(\target)$ instance. 
We suspect that obtaining such a complexity classification is a very hard problem
even for restricted language classes such as equality languages.
One of the reasons for this is the very liberal way that the input is represented.
If one changes the representation, then a complexity classification may be easier to obtain.
A plausible way of doing this is to
restrict ourselves to $\omega$-{\em categorical} base structures.
The concept of $\omega$-categoricity plays a key role in the study of complexity
aspects of CSPs~\cite{Bodirsky:InfDom}, but it is also important from an AI perspective~\cite{hirsch1996relation,Huang:kr2012,Jonsson:ai2018}.
Examples of such structures include all structures with a finite domain and many relevant infinite-domain structures
such as $({\mathbb N};=)$, $({\mathbb Q};<)$, and the standard structures underlying formalisms such as
Allen's algebra and RCC.
For $\omega$-categorical base structures $\base$, each fo-definable relation $R$ can be partitioned into a finite number of equivalence classes with respect to the automorphism group of $\base$, and 
this gives a much more restricted way of representing the input. We leave this as an interesting future research project.

There are several other interesting research directions concerning the computation of simplification maps.
Our definition of simplification maps requires them to be maximally defined, in the sense that the map must be defined if a simplification over the target language exists. 
This is not a problem when both the target and source languages are finite, but if e.g.\ the source language is infinite then the landscape becomes more complicated, in particular if the general
construction problem is undecidable. 
Thus, does it make sense to allow suboptimal simplification maps which are 
oblivious to certain types of constraints, but which can be computed more efficiently? Or simplification maps where not all entries are polynomial time accessible? 
%And is it possible to classify infinite languages admitting efficiently computable simplification maps? 

\paragraph*{Computation of branching maps}

Branching is one of the basic techniques for designing exponential-time algorithms
for combinatorial problems such as the CSP~\cite{Morrison:etal:do2016,Woeginger:CO2003}. Typically, one divide the problem recursively into smaller subproblems, and these
smaller subproblems are often computed by rewriting a small neighbourhood around a variable or a constraint in two or more
ways. This process is then required to preserve solvability in the following way: the original instance is satisfiable if and only
if at least one of the branching instances is satisfiable.
The close connection with branching maps is quite obvious: the small neighbourhood is the induced subinstance with
a particular radius, the rewriting process is carried out with the aid of the branching map, and we require
that this is done in a solution-preserving way (which is a stronger condition than the one needed for branching algorithms).
Given that it is a major research area to find good branching algorithms for combinatorial problems, it
is safe to assume that the computation of branching maps with reasonable branching factors is a very difficult problem, and that reasonable branching maps are highly problem dependent. Thus, it would be interesting to compare the bounds obtained by sidedoors and backdoors to the current best algorithms for other types of CSPs.
%\todo[inline]{PJ: "...for other types of spatiotemporal reasoning problems." $\rightarrow$ "...for various types of CSPs."?}
For which instances is each approach particularly suitable?
%Does there e.g.\ exist real world examples where the sidedoor approach is  %While it seems unlikely that one approach will always outperform the others, it could be interesting to determine classes of instances w

\paragraph*{Integer programming} 
A challenging research direction is to use the backdoor approach for (mixed) integer linear programming ((M)ILP), which
can be seen as an example of infinite CSPs that are not based on a finite set of JEPD relations. One very prominent result in this direction is Lenstra's celebrated
result showing that MILP is fixed-parameter tractable parameterized by the number of integer variables~\cite{Lenstra83}, which can be seen as a (variable) backdoor to LP.
In this context it seems particularly promising to explore variable (or even constraint) backdoors to the recently introduced tractable classes using the treedepth of the primal and dual graph
of an MILP instance~\cite{BrandKO21,abs-1904-01361} as such backdoors have the potential to generalise Lenstra's algorithm. It seems 
unlikely, though, that a general approach (as the one introduced
in this article) will be fruitful for MILP: we believe it is more likely that the backdoor approach for MILP requires tailor-made solutions for each tractable fragment. However, we still believe that our backdoor approach may serve as inspiration for this research direction.
We also believe that sidedoors may be relevant in this context but this is a more speculative idea.

\paragraph*{Satisfiability modulo theories}

{\em Satisfiability modulo theories} (SMT) is a formalism where one is given a first-order formula (with respect to a background theory $\Phi$) and wants to determine whether the formula admits at least one model. 
The reader unfamiliar with this formalism may e.g. consult the introduction by Barrett et al. (2009) \nocite{Barrett:etal:SMT}.
Thus, for a theory $\Phi$ we let SMT$(\Phi)$ be the problem of determining whether a given
first-order $\Phi$-formula
is satisfiable. 
One may discuss what a suitable backdoor concept for SMT is but one plausible idea is
that given such a formula, we extend it with atomic formulas (corresponding to the function $\alpha$ in
the CSP backdoor approach) and rewrite the resulting formula into some tractable class of formulas (corresponding
to the simplification map).
However, this idea immediately makes it possible to view the CSP backdoor problem
for infinite language as a restricted version of the backdoor problem for SMT, and this even holds for the
empty theory which corresponds to equality languages. The consequence of this
is obvious: the hardness example in Section~\ref{sec:inf} leads to \Weft[2]-hardness since every relation
that is used can be computed in polynomial time.
%Sioutis and Janhunen~\cite[Section 4]{Sioutis:Janhunen:ki2019} discuss further aspects of combining SMT with backdoors.

If we look at sidedoors instead, then the difficulties become worse since we now want to rewrite
subformulas of the given input formula. If the formula is written in CNF, then we can most likely
use an approach similar to branching maps. However, requiring the formula to be in CNF is unrealisitc
since it is directly opposed to the idea of SMT being an efficient way of representing
computational problems: naturally, we can convert the input formula to CNF but
this is not efficient since there are DNF formulas whose CNF equivalents are exponentially
larger.

\bibliographystyle{abbrv}
\bibliography{references}

\begin{thebibliography}{10}

\bibitem{Barrett:etal:SMT}
C.~Barrett, R.~Sebastiani, S.~Seshia, and C.~Tinelli.
\newblock Satisfiability modulo theories.
\newblock In {\em Handbook of Satisfiability}, volume 185 of {\em Frontiers in
  Artificial Intelligence and Applications}, pages 825--885. {IOS} Press, 2009.

\bibitem{barto2016}
L.~Barto and M.~Pinsker.
\newblock The algebraic dichotomy conjecture for infinite domain constraint
  satisfaction problems.
\newblock In {\em Proc. 31st Annual ACM/IEEE Symposium on Logic in Computer
  Science (LICS-2016)}, pages 615--622, 2016.

\bibitem{Bodirsky:InfDom}
M.~Bodirsky.
\newblock {\em Complexity of Infinite-Domain Constraint Satisfaction}.
\newblock Cambridge University Press, 2021.

\bibitem{DBLP:journals/jair/BodirskyJ17}
M.~Bodirsky and P.~Jonsson.
\newblock A model-theoretic view on qualitative constraint reasoning.
\newblock {\em Journal of Artificial Intelligence Research}, 58:339--385, 2017.

\bibitem{bodirskykara2010}
M.~Bodirsky and J.~K\'{a}ra.
\newblock The complexity of temporal constraint satisfaction problems.
\newblock {\em Journal of the ACM}, 57(2):9:1--9:41, 2010.

\bibitem{Bodirsky:Wolfl:ijcai2011}
M.~Bodirsky and S.~W{\"{o}}lfl.
\newblock {RCC8} is polynomial on networks of bounded treewidth.
\newblock In {\em Proc. of the 22nd International Joint Conference on
  Artificial Intelligence (IJCAI-2011)}, pages 756--761, 2011.

\bibitem{BrandKO21}
C.~Brand, M.~Kouteck{\'{y}}, and S.~Ordyniak.
\newblock Parameterized algorithms for {MILP}s with small treedepth.
\newblock In {\em Proc. 35th AAAI Conference on Artificial Intelligence
  (AAAI-2021)}, pages 12249--12257, 2021.

\bibitem{bulatov2017}
A.~Bulatov.
\newblock A dichotomy theorem for nonuniform {CSP}s.
\newblock In {\em Proc. 58th Annual Symposium on Foundations of Computer
  Science ({FOCS-2017})}, pages 319--330, 2017.

\bibitem{Carbonnel:Cooper:constraints2016}
C.~Carbonnel and M.~Cooper.
\newblock Tractability in constraint satisfaction problems: a survey.
\newblock {\em Constraints}, 21(2):115--144, 2016.

\bibitem{10.1007/978-3-319-10428-7_18}
C.~Carbonnel, M.~C. Cooper, and E.~Hébrard.
\newblock On backdoors to tractable constraint languages.
\newblock In {\em Proc. 20th International Conference on Principles and
  Practice of Constraint Programming (CP-2014)}, pages 224--239, 2014.

\bibitem{rod1999parameterized}
R.~Downey and M.~Fellows.
\newblock {\em Parameterized complexity}.
\newblock Monographs in Computer Science. Springer, 1999.

\bibitem{Dvorak:etal:ai2012}
W.~Dvor{\'{a}}k, S.~Ordyniak, and S.~Szeider.
\newblock Augmenting tractable fragments of abstract argumentation.
\newblock {\em Artificial Intelligence}, 186:157--173, 2012.

\bibitem{Dylla:2017:SQS:3058791.3038927}
F.~Dylla, J.~Lee, T.~Mossakowski, T.~Schneider, A.~V. Delden, J.~V.~D. Ven, and
  D.~Wolter.
\newblock A survey of qualitative spatial and temporal calculi: Algebraic and
  computational properties.
\newblock {\em ACM Computing Surveys}, 50(1):7:1--7:39, 2017.

\bibitem{abs-1904-01361}
F.~Eisenbrand, C.~Hunkenschr{\"{o}}der, K.~Klein, M.~Kouteck{\'{y}}, A.~Levin,
  and S.~Onn.
\newblock An algorithmic theory of integer programming.
\newblock {\em CoRR}, abs/1904.01361, 2019.

\bibitem{FICHTE201564}
J.~Fichte and S.~Szeider.
\newblock Backdoors to tractable answer set programming.
\newblock {\em Artificial Intelligence}, 220:64 -- 103, 2015.

\bibitem{book/FlumG06}
J.~Flum and M.~Grohe.
\newblock {\em Parameterized Complexity Theory}.
\newblock Springer, 2006.

\bibitem{GanianRS17}
R.~Ganian, M.~S. Ramanujan, and S.~Szeider.
\newblock Discovering archipelagos of tractability for constraint satisfaction
  and counting.
\newblock {\em {ACM} Transactions on Algorithms}, 13(2):29:1--29:32, 2017.

\bibitem{gj79}
M.~Garey and D.~Johnson.
\newblock {\em Computers and intractability: a guide to the theory of
  {NP}-completeness}.
\newblock W.H. Freeman and Company, 1979.

\bibitem{GASPERS201738}
S.~Gaspers, N.~Misra, S.~Ordyniak, S.~Szeider, and S.~\v{Z}ivn\'{y}.
\newblock Backdoors into heterogeneous classes of {SAT} and {CSP}.
\newblock {\em Journal of Computer and System Sciences}, 85:38--56, 2017.

\bibitem{Gaspers:etal:backdoor}
S.~Gaspers, S.~Ordyniak, and S.~Szeider.
\newblock Backdoor sets for {CSP}.
\newblock In {\em The Constraint Satisfaction Problem: Complexity and
  Approximability}, volume~7 of {\em Dagstuhl Follow-Ups}, pages 137--157.
  2017.

\bibitem{gaspersS12}
S.~Gaspers and S.~Szeider.
\newblock Backdoors to satisfaction.
\newblock In {\em The Multivariate Algorithmic Revolution and Beyond - Essays
  Dedicated to Michael R. Fellows on the Occasion of His 60th Birthday}, pages
  287--317. Springer, 2012.

\bibitem{Goerdt:csr2009}
A.~Goerdt.
\newblock On random ordering constraints.
\newblock In {\em Proc. 4th International Computer Science Symposium in Russia
  (CSR-2009)}, volume 5675, pages 105--116, 2009.

\bibitem{Guruswami:etal:sicomp2011}
V.~Guruswami, J.~H{\aa}stad, R.~Manokaran, P.~Raghavendra, and M.~Charikar.
\newblock Beating the random ordering is hard: Every ordering {CSP} is
  approximation resistant.
\newblock {\em {SIAM} Journal on Computing}, 40(3):878--914, 2011.

\bibitem{Guttmann:Maucher:ifiptcs2006}
W.~Guttmann and M.~Maucher.
\newblock Variations on an ordering theme with constraints.
\newblock In {\em Proc. 4th {IFIP} International Conference on Theoretical
  Computer Science {(TCS}-2006)}, volume 209, pages 77--90, 2006.

\bibitem{hirsch1996relation}
R.~Hirsch.
\newblock Relation algebras of intervals.
\newblock {\em Artificial intelligence}, 83(2):267--295, 1996.

\bibitem{HodgesLong}
W.~Hodges.
\newblock {\em Model theory}.
\newblock Cambridge University Press, 1993.

\bibitem{Holyer:sicomp81}
I.~Holyer.
\newblock The {NP}-completeness of some edge-partition problems.
\newblock {\em {SIAM} Journal on Computing}, 10(4):713--717, 1981.

\bibitem{Huang:kr2012}
J.~Huang.
\newblock Compactness and its implications for qualitative spatial and temporal
  reasoning.
\newblock In {\em Proc. 13th International Conference on Principles of
  Knowledge Representation and Reasoning (KR-2012)}, 2012.

\bibitem{Jonsson:ai2018}
P.~Jonsson.
\newblock Constants and finite unary relations in qualitative constraint
  reasoning.
\newblock {\em Artificial Intelligence}, 257:1--23, 2018.

\bibitem{lagerkvist2017d}
P.~Jonsson and V.~Lagerkvist.
\newblock An initial study of time complexity in infinite-domain constraint
  satisfaction.
\newblock {\em Artificial Intelligence}, 245:115--133, 2017.

\bibitem{Jonsson:Lagerkvist:mfcs2018}
P.~Jonsson and V.~Lagerkvist.
\newblock Why are {CSP}s based on partition schemes computationally hard?
\newblock In {\em Proc. 43rd International Symposium on Mathematical
  Foundations of Computer Science (MFCS-2018)}, pages 43:1--43:15, 2018.

\bibitem{Jonsson:etal:cp2021}
P.~Jonsson, V.~Lagerkvist, and S.~Ordyniak.
\newblock Reasoning short cuts in infinite domain constraint satisfaction:
  Algorithms and lower bounds for backdoors.
\newblock In {\em Proc. 27th International Conference on Principles and
  Practice of Constraint Programming (CP-2021)}, pages 32:1--32:20, 2021.

\bibitem{lagerkvist2022e}
P.~Jonsson, V.~Lagerkvist, and S.~Ordyniak.
\newblock Computational short cuts in infinite domain constraint satisfaction.
\newblock {\em Journal of Artificial Intelligence Research}, 75, 2022.

\bibitem{Jonsson:etal:ai2021}
P.~Jonsson, V.~Lagerkvist, and G.~Osipov.
\newblock Acyclic orders, partition schemes and {CSP}s: Unified hardness proofs
  and improved algorithms.
\newblock {\em Artificial Intelligence}, 296:103505, 2021.

\bibitem{kavvadias98}
D.~Kavvadias and M.~Sideri.
\newblock The inverse satisfiability problem.
\newblock {\em SIAM Journal on Computing}, 28:152--163, 1998.

\bibitem{backdoors2005}
P.~Kilby, J.~Slaney, S.~Thi\'{e}baux, and T.~Walsh.
\newblock Backbones and backdoors in satisfiability.
\newblock In {\em Proc. 20th National Conference on Artificial Intelligence
  (AAAI-2005)}, page 1368–1373, 2005.

\bibitem{Kronegger:etal:ai2019}
M.~Kronegger, S.~Ordyniak, and A.~Pfandler.
\newblock Backdoors to planning.
\newblock {\em Artificial Intelligence}, 269:49--75, 2019.

\bibitem{lagerkvist2020e}
V.~Lagerkvist and B.~Roy.
\newblock Complexity of inverse constraint problems and a dichotomy for the
  inverse satisfiability problem.
\newblock {\em Journal of Computer and System Sciences}, 117:23--39, 2021.

\bibitem{Lenstra83}
H.~W. {Lenstra, Jr.}
\newblock Integer programming with a fixed number of variables.
\newblock {\em Mathematics of Operations Research}, 8(4):538--548, 1983.

\bibitem{Ligozat:Renz:pricai2004}
G.~Ligozat and J.~Renz.
\newblock What is a qualitative calculus? {A} general framework.
\newblock In {\em Proc. 8th Pacific Rim International Conference on Artificial
  Intelligence (PRICAI-2004)}, pages 53--64, 2004.

\bibitem{Meier:etal:algo2019}
A.~Meier, S.~Ordyniak, M.~S. Ramanujan, and I.~Schindler.
\newblock Backdoors for linear temporal logic.
\newblock {\em Algorithmica}, 81(2):476--496, 2019.

\bibitem{Mohring:etal:sicomp2004}
R.~H. M{\"{o}}hring, M.~Skutella, and F.~Stork.
\newblock Scheduling with {AND/OR} precedence constraints.
\newblock {\em {SIAM} Journal on Computing}, 33(2):393--415, 2004.

\bibitem{Morrison:etal:do2016}
D.~Morrison, S.~Jacobson, J.~Sauppe, and E.~Sewell.
\newblock Branch-and-bound algorithms: {A} survey of recent advances in
  searching, branching, and pruning.
\newblock {\em Discrete Optimization}, 19:79--102, 2016.

\bibitem{Nebel:Burckert:jacm95}
B.~Nebel and H.~B{\"{u}}rckert.
\newblock Reasoning about temporal relations: {A} maximal tractable subclass of
  allen's interval algebra.
\newblock {\em Journal of the {ACM}}, 42(1):43--66, 1995.

\bibitem{OrdyniakSS21}
S.~Ordyniak, A.~Schidler, and S.~Szeider.
\newblock Backdoor {DNF}s.
\newblock In {\em Proc. 30th International Joint Conference on Artificial
  Intelligence (IJCAI-2021)}, pages 1403--1409, 2021.

\bibitem{Pfandler:etal:ijcai2013}
A.~Pfandler, S.~R{\"{u}}mmele, and S.~Szeider.
\newblock Backdoors to abduction.
\newblock In {\em Proc. 23rd International Joint Conference on Artificial
  Intelligence (IJCAI-2013)}, pages 1046--1052, 2013.

\bibitem{DBLP:journals/ai/RenzN99}
J.~Renz and B.~Nebel.
\newblock On the complexity of qualitative spatial reasoning: {A} maximal
  tractable fragment of the region connection calculus.
\newblock {\em Artificial Intelligence}, 108(1-2):69--123, 1999.

\bibitem{Samer:Szeider:aaai2008}
M.~Samer and S.~Szeider.
\newblock Backdoor trees.
\newblock In {\em Proc. 23rd {AAAI} Conference on Artificial Intelligence
  (AAAI-2008)}, pages 363--368, 2008.

\bibitem{Samer:Szeider:jar2009}
M.~Samer and S.~Szeider.
\newblock Backdoor sets of quantified boolean formulas.
\newblock {\em Journal of Automated Reasoning}, 42(1):77--97, 2009.

\bibitem{Samer:Szeider:HoS}
M.~Samer and S.~Szeider.
\newblock Fixed-parameter tractability.
\newblock In {\em Handbook of Satisfiability}, volume 185 of {\em Frontiers in
  Artificial Intelligence and Applications}, pages 425--454. {IOS} Press, 2009.

\bibitem{Sioutis:Janhunen:ki2019}
M.~Sioutis and T.~Janhunen.
\newblock Towards leveraging backdoors in qualitative constraint networks.
\newblock In {\em Proc. 42nd German Conference on AI (KI-2019)}, pages
  308--315, 2019.

\bibitem{Vilain:Kautz:aaai86}
M.~Vilain and H.~Kautz.
\newblock Constraint propagation algorithms for temporal reasoning.
\newblock In {\em Proc. 5th National Conference on Artificial Intelligence
  (AAAI-1986)}, pages 377--382, 1986.

\bibitem{williams2003}
R.~Williams, C.~Gomes, and B.~Selman.
\newblock Backdoors to typical case complexity.
\newblock In {\em Proc. 18th International Joint Conference on Artificial
  Intelligence (IJCAI-2003)}, pages 1173--1178, 2003.

\bibitem{Woeginger:CO2003}
G.~Woeginger.
\newblock Exact algorithms for {NP}-hard problems: a survey.
\newblock In M.~Juenger, G.~Reinelt, and G.~Rinaldi, editors, {\em
  Combinatorial Optimization -- Eureka! You Shrink!}, pages 185--207, 2000.

\bibitem{zhuk20}
D.~Zhuk.
\newblock A proof of the {CSP} dichotomy conjecture.
\newblock {\em Journal of the {ACM}}, 67(5):30:1--30:78, 2020.

\end{thebibliography}

\end{document}